\newif\ifdouble

\doubletrue

\ifdouble
\documentclass[journal]{IEEEtran}
\else
\documentclass[12pt, draftclsnofoot, onecolumn]{IEEEtran}
\fi

\usepackage{amsthm}
\usepackage{latexsym}
\usepackage{amssymb}
\usepackage{graphicx}
\usepackage{caption}
\usepackage{subcaption}
\usepackage{amsmath}
\usepackage{color}
\usepackage{cite}
\usepackage{cleveref}
\usepackage[english]{babel}

\theoremstyle{plain}
\newtheorem{theorem}{Theorem}

\newtheorem{proposition}{Proposition}
\newtheorem{corollary}{Corollary}
\theoremstyle{definition}
\newtheorem{definition}{Definition}
\theoremstyle{remark}
\newtheorem{remark}{Remark}

\crefformat{footnote}{#2\footnotemark[#1]#3}
\ifCLASSINFOpdf
\else
\fi

\begin{document}

\title{Secure Multi-Source Multicast
\thanks{A.Cohen, A,Cohen and O. Gurewitz are with the Department of Communication Systems Engineering, Ben-Gurion University of the Negev, Beer-Sheva 84105, Israel (e-mail: alejandr@post.bgu.ac.il; coasaf@post.bgu.ac.il; gurewitz@post.bgu.ac.il).
M.M\'{e}dard is with the Laboratory for Information and Decision Systems at the Massachusetts Institute of Technology (medard@mit.edu). This research was partially supported by the Israeli MOITAL NEPTUN consortium and by the European Union Horizon 2020 Research and Innovation Programme SUPERFLUIDITY under Grant 671566. Parts of this work appeared at the IEEE International Symposium on Information Theory (ISIT), 2017.}
}
\markboth{}{}
\ifdouble
\author{\IEEEauthorblockN{Alejandro Cohen\hspace{15 mm} Asaf Cohen \hspace{15 mm} Muriel M\'{e}dard \hspace{15 mm} Omer Gurewitz}\\ \hspace{1mm} BGU \hspace{28mm} BGU \hspace{29mm} MIT \hspace{31mm} BGU \vspace{-8mm}}
\else
\author{\IEEEauthorblockN{Alejandro Cohen\hspace{12 mm} Asaf Cohen \hspace{12 mm} Muriel M\'{e}dard \hspace{12 mm} Omer Gurewitz}\\ \hspace{1mm} BGU \hspace{28mm} BGU \hspace{29mm} MIT \hspace{31mm} BGU \vspace{-8mm}}
\fi
\maketitle
\begin{abstract}
The principal mission of \emph{Multi-Source Multicast} (MSM) is to disseminate all messages from all sources in a network to all destinations. MSM is utilized in numerous applications. In many of them, securing the messages disseminated is critical.

A common secure model is to consider a network where there is an eavesdropper which is able to observe a subset of the network links, and seek a code which keeps the eavesdropper ignorant regarding \emph{all the messages}. While this is solved when all messages are located at a single source, \emph{Secure MSM} (SMSM) is an open problem, and the rates required are hard to characterize in general.

In this paper, we consider \emph{Individual Security}, which promises that the eavesdropper has zero mutual information with \emph{each message individually}, or, more generally, with \emph{sub sets of messages}. We completely characterize the rate region for SMSM under individual security, and show that such a security level is achievable at the full capacity of the network, that is, the cut-set bound is the matching converse, similar to \emph{non-secure} MSM.
Moreover, we show that the field size is similar to non-secure MSM and does not have to be larger due to the security constraint.
\end{abstract}

\section{Introduction}\label{intro}
Linear Network Coding (LNC) \cite{li2003linear} and Random Linear Network Coding (RLNC) \cite{ho2006random} are essential for efficient utilization of network resources. With network coding, \emph{multiple sources} can multicast information to all destinations simultaneously, at rates up to the min-cut between the sources and the destinations. \Cref{fig:wiretap_Secure gossip} depicts a simple example: the min-cut from any source to any destination is 2, and from both sources to any destination is 4, hence one can disseminate \emph{2 messages from each source to all destinations}. However, in many practical multicast applications, it is important to ensure privacy is not compromised if an eavesdropper (Eve) is present in the network. Indeed, the theory of secure network coding is vast. We include here only the most relevant works.

When the sources are co-located at a single node, several secure network coding solutions were suggested \cite{cai2002secure,chan2008capacity,cai2011secure,el2007wiretap,silva2008security,el2012secure}. Such solutions guarantee the mutual information between Eve's data, $\textbf{Z}$, and all the messages is 0. For example, returning to \Cref{fig:wiretap_Secure gossip}, if only source $s_1$ had messages to send, and Eve would be able to wiretap one link in the network, then secure network coding would guarantee secure dissemination of one message from the source to all destinations. This is a reduction in rate compared to the full capacity, as the min-cut from $s_1$ to any destination is 2. However, when requiring zero mutual information with all messages from the source, this rate reduction is essential, and matches the converse result.
\ifdouble
\begin{figure}
\centering
\includegraphics[trim=0cm 0.0cm 0cm 0cm,clip,scale=0.85]{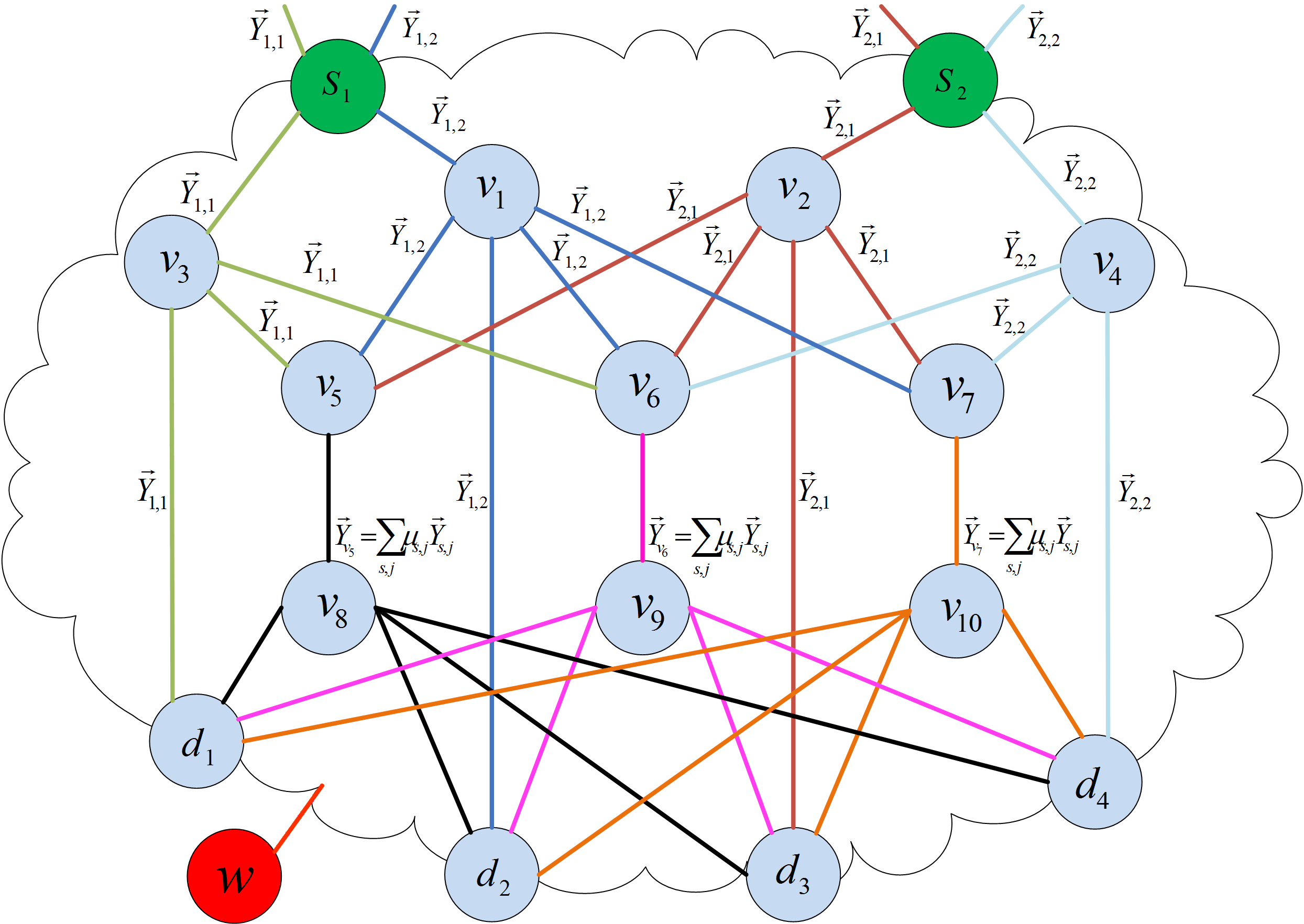}
\caption{Secure multi-source multicast with LNC, for two sources $s_i$, with two messages each and four legitimate destination nodes $d_i$. The eavesdropper min-cut is at most $1$. The edges in the graph point downward.}
\label{fig:wiretap_Secure gossip}
\end{figure}
\else
\begin{figure}
\centering
\includegraphics[trim=0cm 0.0cm 0cm 0cm,clip,scale=1.10]{ex1_2.png}
\caption{Secure multi-source multicast with LNC, for two sources $s_i$, with two messages each and four legitimate destination nodes $d_i$. The eavesdropper min-cut is at most $1$. The edges in the graph point downward.}
\label{fig:wiretap_Secure gossip}
\end{figure}
\fi

When the network includes multiple sources which are not co-located, the problem is more involved. Clearly, applying a single-source, secure network coding solution at each source would give an achievable scheme. In the example, if Eve wiretaps one link, one can clearly multicast one message \emph{from each source, to all destinations}. This solution, however, may be wasteful, as it is half of the full capacity of the network, ``wasting" one message \emph{per source}, although Eve may capture only a single link regardless of the number of sources. Indeed, there is no matching converse result for the above solution.

In \cite{cai2007security,zhang2009general}, the authors gave a necessary and sufficient condition for Secure Multi-Source Multicast (SMSM). However, it is a condition on \emph{ranks of matrices} having the global encoding vectors as columns, and, unlike non-secure MSM or secure single-source multicast, it does not translate directly to \emph{rate or min-cut constraints}. Thus, the problem of determining the rate region in SMSM is an open problem in general \cite{cai2009valuable}, and as mentioned in \cite[Section VI]{cai2011theory}, seeking models for which it is solvable is important.
In \cite{chan2014network}, the authors characterized the network coding capacity of several models, including SMSM, via the entropic region $\Gamma^*$. Yet, to date, this region is not fully characterized.

\subsection*{Main Contribution}
In this paper, we consider SMSM under an \emph{Individual Security} constraint.
In this model, the eavesdropper is kept ignorant, in the sense of having zero mutual information, regarding each message separately (or, more generally, regarding sub sets of messages), yet may potentially obtain \emph{insignificant} information about mixtures of packets transmitted. Such a security model was recently used in various canonical problems, e.g., wiretap channels \cite{kobayashi2013secure}, more general broadcast channels \cite{mansour2014secrecy,chen2015individual,mansour2015individual,mansourindividual} and multiple-access channels \cite{goldenbaum2015multiple,chensecure}, and, although not specifically mentioned as such, is also related to weakly secure network coding \cite{bhattad2005weakly} and the notion of \emph{algebraic security} \cite{lima2007random,claridge2017probability}, which consider the information in linear combinations of messages. Moreover, a related single-source problem is that of distributed storage \cite{kadhe2014weaklyNetCod,kadhe2014weakly,paunkoska2016improved}, which we also address.

We completely characterize the rate region for individually secure MSM. Specifically, we show that secure communication is achievable up to the min-cut, that is, without any decrease in the rate or any message ``blow-up" by extra randomness. In fact, due to the individual security constraint, messages protect one another, and in the context of \Cref{fig:wiretap_Secure gossip}, one is able to send \emph{two messages from each source securely, although Eve may observe any single link}. In that sense, we non-trivially extend the single-source multicast results of \cite{bhattad2005weakly} and \cite{silva2009universal} to multi-source multicast, giving both linear codes as well as non-linear codes over a small field size.

We then turn to a few applications where the suggested coding scheme can be useful. Specifically, we consider data centers, wireless networks and live broadcasting of video using multi-path streaming, and show how the individual security coding schemes suggested in this paper are applicable, achieving the full capacity of those systems.
Finally, we show that the coding scheme is applicable to algebraic gossip as well \cite{deb2006algebraic}, resulting in \emph{secure gossip} without extra rounds. For example, consider the "Random Phone Call" model. This model was introduced in \cite{demers1987epidemic} as special case of uniform gossip. In each round of communication, every participant may "call" a random participant, and send one unit of information. The goal is, naturally, to disseminate messages from the source to \emph{all} participants. Rigorously, the underlying graph is complete and unweighted. A detailed analysis of this model is given in \cite{deb2005random,haeupler2011analyzing}. It was shown that in a random phone call model with $v$ nodes, the flooding time is $\Theta(\log v)$, with constant throughput. Of course, this is without any secrecy constraint. Any phone call which Eve listens to contains relevant information, and results in leakage. Using the code suggested in this paper, we will show that one can design a secure gossip scheme, which makes sure that as long as Eve does not listen to too many calls, she remains complectly ignorant regarding any specific message, and all this without any loss in throughput or number of rounds.

The structure of this paper is as follows. In \Cref{formulationMultiplex}, a SMSM model is formally described. \Cref{main results} includes our main results, with the individually-SMSM achievability proved in \Cref{SecureGossipAlgorithmMultiplex} and converse proved in \Cref{converse_sec}. \Cref{LinearCodes} includes a linear code construction for the individually-SMSM model. \Cref{strong-SMSM} describes a Strongly-SMSM algorithm and proves a direct result for it. In \Cref{applications}, we show a few important examples, for which the individual security coding is applicable. \Cref{conc} concludes the paper.

\section{Model and Problem Formulation}\label{formulationMultiplex}
SMSM is specified by a graph $\mathcal{G}=(\mathcal{V},\mathcal{E})$, where $\mathcal{V}$ and $\mathcal{E}$ are the node set and the edge set, respectively. We assume noise-free links of unit capacity. This capacity can be thought of as one "packet" of $c$ bits, plus some negligible overhead.\footnote{\label{note1}As in most LNC solutions, a header is required for each message. Thus, we assume messages of length $c$, large enough to make the overhead in the header negligible.}

The node set $\mathcal{V}$ contains a subset of source nodes $\mathcal{S} = \{S_1,\ldots S_{|S|}\}$ and a subset of legitimate destination nodes $\mathcal{D}= \{D_1,\ldots D_{|D|}\}$. Each of the sources has its own set of $k$ independent and uniformly distributed messages of length $c$ each, over the binary field. We denote them by a messages matrix
\begin{equation*}
   \textbf{M}_s=[\vec{M}_{s,1};\vec{M}_{s,2};\ldots; \vec{M}_{s,k}] \in \{0,1\}^{k\times c},
\end{equation*}
where each row corresponds to a separate message $\vec{M}_{s,j}$, $j\in\{1,\ldots,k\}$. Note that both the independence of the messages, as well as their uniform distribution are critical to achieve secrecy. These assumption are, indeed, common in the related literature as well \cite{chensecure,silva2009universal,cai2011theory}.

We assume an eavesdropper which can obtain a subset of $w$ packets traversing the network. Specifically, we define the eavesdropper matrix as
\begin{equation*}
     \textbf{Z}_{w}=[Z_{1}^{c};Z_{2}^{c};\ldots; Z_{w}^{c}] \in \{0,1\}^{w\times c}.
\end{equation*}

We denote the values of min-cuts in the network by $\rho(.;.)$.
For example, for $s_1\in\mathcal{S}$ and $d_1\in\mathcal{D}$, $\rho(s_1;d_1)$ represents the value of the min-cut from source node $s_1$ to legitimate node $d_1$. $\rho(s_1;z)$ represents the value of the min-cut from source node $s_1$ to the eavesdropper (assuming $z$ is a virtual node with infinite capacity from the $w$ edges observed by Eve) and $\rho(\mathcal{S};d_1)$ represents the value of the min-cut from all the source nodes to legitimate node $d_1$.

The goal is to design secure multi-source multicast coding scheme where legitimate nodes send their available messages in order to disseminate all the messages to all the legitimate destination nodes, yet, observing $w$ packets from the communication between legitimate nodes, the eavesdropper is ignorant regarding the messages.
\begin{definition}\label{constraints}
An MSM algorithm with parameters $k$ and $w$ is \emph{Reliable} and \emph{Individually} or \emph{Strongly} secure if:\\
(1) Reliable: At the legitimate destination node $d\in\mathcal{D}$, letting $\textbf{Y}_d$ denote the message matrix obtained, for any set of messages $\textbf{M}_s, s\in \mathcal{S}$, we have
\begin{equation*}
     P(\hat{\textbf{M}}_s(\textbf{Y}_d) \ne \textbf{M}_s) \leq \epsilon,
\end{equation*}
where $\hat{\textbf{M}}_s(\textbf{Y}_d)$ is the estimation of messages $\textbf{M}_s$ at $d$.

\noindent (2) Individually secure: At the eavesdropper, observing $w$ packets, we have
\begin{equation*}
     H(M_{s,j}|\textbf{Z}_{w})=H(M_{s,j}),
\end{equation*}
for all $j \in \{1,\ldots,k\}$ and for all $s\in\mathcal{S}$.

\noindent (3) Strongly secure: At the eavesdropper, observing $w$ packets, for all $s\in\mathcal{S}$ we have
\begin{equation*}
  H(\textbf{M}_{s}|\textbf{Z}_{w})=H(\textbf{M}_{s}).
\end{equation*}
\end{definition}
\begin{remark}
The individual-secrecy constraint given in \Cref{constraints}.2 does not promise perfect, strong-secrecy \cite{cai2011secure,cai2007security,el2012secure}, which is, having the mutual information with all messages negligible. Individual-secrecy ensures secrecy only on each message $M_{s,j}$ separately. The eavesdropper, observing $\textbf{Z}_{w}$, may obtain some information on the combination of $k$ messages since the messages are not independent given $\textbf{Z}_{w}$. However, since the $k$ original messages are mutually independent, the leaked information has no meaning \cite{kobayashi2013secure, mansour2014secrecy,chen2015individual,mansour2015individual,mansourindividual, goldenbaum2015multiple,chensecure, mansour2015individual1}.
In other words, if the messages are independent, we have
\ifdouble
\begin{align*}
I(M_{s,k};\textbf{Z}_{w} & | M_{s,1},\ldots,M_{s,k-1}) \\
&= H(M_{s,k}|M_{s,1},\ldots,M_{s,k-1}) \\
& \qquad -H(M_{s,k}|\textbf{Z}_{w},M_{s,1},\ldots,M_{s,k-1}) \\
&=H(M_{s,k})-H(M_{s,k}|\textbf{Z}_{w},M_{s,1},\ldots,M_{s,k-1}) \\
&\ge H(M_{s,k})-H(M_{s,k}|\textbf{Z}_{w}) \\
&= I(M_{s,k};\textbf{Z}_{w}).
\end{align*}
\else
\begin{eqnarray*}
I(M_{s,k};\textbf{Z}_{w}|M_{s,1},\ldots,M_{s,k-1}) &=& H(M_{s,k}|M_{s,1},\ldots,M_{s,k-1})-H(M_{s,k}|\textbf{Z}_{w},M_{s,1},\ldots,M_{s,k-1}) \\
&=&H(M_{s,k})-H(M_{s,k}|\textbf{Z}_{w},M_{s,1},\ldots,M_{s,k-1}) \\
& \ge &H(M_{s,k})-H(M_{s,k}|\textbf{Z}_{w}) \\
&=& I(M_{s,k};\textbf{Z}_{w}).
\end{eqnarray*}
\fi
Hence,
\ifdouble
\begin{multline*}
 I(\textbf{M}_{s};\textbf{Z}_{w}) = \sum_{k}I(M_{s,k};\textbf{Z}_{w}|M_{s,1},\ldots,M_{s,k-1})
\\
\geq \sum_{k}I(M_{s,k};\textbf{Z}_{w}).
\end{multline*}
\else
\[
 I(\textbf{M}_{s};\textbf{Z}_{w}) = \sum_{k}I(M_{s,k};\textbf{Z}_{w}|M_{s,1},\ldots,M_{s,k-1})\geq \sum_{k}I(M_{s,k};\textbf{Z}_{w}).
\]
\fi
We require that the r.h.s will be small, however, this does not guarantee that the l.h.s is small. If the eavesdropper receives message $M_{s,j}$ by any other manner than the Individual-SMSM transmissions, Eve may obtain some information on other messages $M_{s,i}, i \neq j$, from $M_{s,j}$ and $\textbf{Z}_{w}$. If it is required to prevent the possibility of such an attack, one can get perfect secrecy using  \Cref{constraints}.3, yet at the price of a lower rate, as given in \Cref{strong-SMSM}.
\end{remark}
\begin{remark}
For multicast problems and LNC, the condition in $(1)$ can be used with $\epsilon=0$ \cite{li2003linear,ho2006random}. Yet, we allow a small error to cope with protocols such as randomized gossip \cite{deb2006algebraic,cohen2015network}, which we discuss later in this paper.
\end{remark}
\begin{remark}
The first code construction we consider, given in \Cref{SecureGossipAlgorithmMultiplex}, is based on random coding. Therefore, in that case, the individual secrecy constraint will hold only asymptotically, that is,
\[
H(M_{s,j}|\textbf{Z}_{w})/H(M_{s,j})\rightarrow 1
\]
as $k$ grows. Then, in \Cref{LinearCodes}, we suggest a structured linear code, which results in zero mutual information, such that there is no requirement for $k$ to grow.
\end{remark}

\subsection{Source and Network Coding}
We assume a source $s\in \mathcal{S}$ may use an encoder,
\begin{equation*}
   f : \mathcal{M}_s \rightarrow \mathcal{X}_s\in\{0,1\}^{n\times c},
\end{equation*}
which maps each message matrix $\textbf{M}_s$ to a matrix $\textbf{X}_s$ of codewords. When using a strong security constraints, e.g., \cite{cai2011secure,el2012secure}, $n>k$ and this represents a message ``blow-up" using a random key, used to confuse Eve. However, the main contribution herein, is that \emph{under individual-secrecy, $n=k$ suffices, and there will be no rate loss due to the secrecy constraint}.

Then, the source packets $\vec{Y}$ transmitted are linear combinations of $\{\vec{X}_{r}\}_{r=1}^{n}$ with coefficients in the usual LNC sense, i.e.,
\begin{equation*}
     \vec{Y}=\sum_{r=1}^{n}\mu_{r}\vec{X}_{r}.
\end{equation*}
Each node maintains a subspace $Y_v$ that is the span of all packets known to it. In RLNC, when node $v$ sends a packet, $Out(\vec{Y})$, it chooses uniformly a packet from $Y_v$ by taking a random linear combination. If a deterministic algorithm is used, e.g., \cite{jaggi2005polynomial}, the coefficients are calculated based on the network topology. The code we suggest herein is only at the sources, and then utilizes any capacity-achieving, non-secure network code.

\subsection{Gossip in Oblivious Networks}\label{Pre}
While the results in this paper are tailored to LNC in the sense of \cite{li2003linear,ho2006random}, they easily apply to \emph{algebraic gossip} \cite{deb2006algebraic} as well. Such algebraic gossip protocol have been considered in the literature for many tasks, such as ensuring database consistency, computing aggregate information and other functions of the data \cite{demers1987epidemic,karp2000randomized,kempe2003gossip,boyd2006randomized}. We briefly describe this model.
The network operates in rounds. In each round $t$, the sources, as well as any legitimate node which has messages it previously received, pick a random node to exchange information with. The information exchange is done by either sending (PUSH) or receiving (PULL) a message. In algebraic gossip, the message sent by a node $v$ is simply a random linear combination of the vectors which form a basis for $Y_v$. The process stops when all the legitimate nodes have all the messages, i.e., have a full rank matrix. We briefly review the definitions and results from \cite{cohen2015network} for non-secure gossip networks, which we will use to formulate our result in this context.
\begin{definition}\label{PreDef2}
A network is \emph{oblivious} if the topology of the network, $G_t$ at time $t$, only depends on $t$, $G_{t^{\prime}}$ for any $t^{\prime} < t$ and some randomness. We call an oblivious network model furthermore i.i.d., if the topology $G_t$ is independent of $t$ and prior topologies.
\end{definition}
The importance of \Cref{PreDef2} lies in the fact that the topology of an oblivious network may change in time, but only based on the past topology and some external randomness. Topology does not change based on the data traversing the network.
Consider a single (uncoded) message, and the set of nodes $S_l$ which received that message after $l$ rounds. $S_l$ advances like a flooding process $F$. That is, $S_l \subseteq S_{l^{\prime}}\subseteq\mathcal{V}$ for $l\leq l^{\prime}$, with an absorbing state $\mathcal{V}$.
We say that $F$ stops at time $t$ if the message is received at all nodes after $t$ rounds.
Let $S_F$ be the random variable denoting the stopping time of $F$.
\begin{definition}\label{PreDef3}
We say an oblivious network with a vertex set $V$ floods in time $T$ with throughput $\alpha$ if there exists a prime power $q$ such that for every vertex $v \in V$ and every $k>0$ we have
\[
P[S_{F} \geq T + k] < q^{-\alpha k}.
\]
\end{definition}

\section{Main Results}\label{main results}
The three main results in this paper completely characterize the rate region for individually secure multi-source multicast.
We give tight achievability and converse, and a tight characterization of the number of rounds required under a gossip model.
Specifically, we first note that the individually secrecy constrain in \Cref{constraints} is $I(M_{s,j}; \textbf{Z}_w)=0$ for any single message $j$. However, ensuring the mapping from $\textbf{M}_s$ to $\textbf{X}_s$ mixes the messages appropriately, i.e., satisfies rank constraints similar to \cite[Lemma 3.1]{cai2011theory}, can, in fact, ensure Eve is kept ignorant on any \emph{set of $k-(w+k\epsilon)$ messages}, where $k\epsilon\geq 1$ is an integer and $\epsilon=o(k)$. That is, guarantee $k_s$-individual perfect secrecy with respect to any set of $k_s\leq k-(w+k\epsilon)$ messages. Let $\textbf{M}_{s}^{k_s}$ denote a set of $k_s$ messages from $s$.
Thus, the first main result is the following achievability theorem, which states that $k_s$-individually-secure multi-source multicast is achievable at rates up to the network min-cuts, using LNC.
\subsection{Individually Secure MSM}
\begin{theorem}\label{direct theorem1}
Assume an SMSM network $(\mathcal{V},\mathcal{E},\mathcal{S},\mathcal{D},w)$.
There exists a coding scheme which disseminates $k$ messages from each source in $\mathcal{S}$, to all destinations in $\mathcal{D}$, while keeping an eavesdropper which observes $w<k$ links ignorant with respect to any set of $k_s \leq k-(w+k\epsilon)$ messages individually, where $\epsilon = o(k)$, such that $I(\textbf{M}^{k_s}_{s};\textbf{Z}_w)=0$, if:
\begin{enumerate}
  \item For all $s\in \mathcal{S}$ and all $d\in \mathcal{D}$, $\rho(s,d) \ge k$.
  \item For all $d\in \mathcal{D}$, $\rho(\mathcal{S},d) \ge k|S|$.
\end{enumerate}
\end{theorem}
In \Cref{L_To_Eve}, we prove the $k_s$-individual perfect secrecy constraint is indeed met. In particular, for any single message $j$ as given in \Cref{constraints}, we have
\begin{corollary}\label{corollary theorem1}
Assume an SMSM network $(\mathcal{V},\mathcal{E},\mathcal{S},\mathcal{D},w)$.
There exists a coding scheme which disseminates $k$ messages from each source in $\mathcal{S}$, to all destinations in $\mathcal{D}$, while keeping an eavesdropper which observes $w<k$ links ignorant with respect to each message individually if:
\begin{enumerate}
  \item For all $s\in \mathcal{S}$ and all $d\in \mathcal{D}$, $\rho(s,d) \ge k$.
  \item For all $d\in \mathcal{D}$, $\rho(\mathcal{S},d) \ge k|S|$.
\end{enumerate}
\end{corollary}

Note that under strong-secrecy, i.e., requiring Eve's mutual information \emph{with all messages simultaneously} to be zero, the problem of MSM is still open \cite{cai2009valuable},\cite[Section VI]{cai2011theory}. Clearly, if Eve observes $w$ links, a naive implementation, which increases the message rates from each source by $w$, can send $k$ messages from each source where $n \geq k + w + k\epsilon$ and achieve strong secrecy if:
\begin{enumerate}
  \item For all $s\in \mathcal{S}$ and all $d\in \mathcal{D}$, $\rho(s,d) \geq n$.
  \item  For all $d\in \mathcal{D}$, $\rho(S,d) \geq n|S|$.
\end{enumerate}
However, such an implementation is clearly wasteful, and, to date, the optimal strategy is unknown. Obviously, the required rates under strong secrecy are higher than the min-cut bound, as even for single-source multicast one needs $\rho(s_1,d_i) \geq k+w$ \cite{cai2011secure}.
Thus, the importance of \Cref{direct theorem1} is that under individual secrecy, not only the rate region can be characterized, and is achievable using linear network coding, individually secure MSM \emph{is possible up to the min-cuts in the network}.

In \Cref{strong-SMSM}, we do provide a code for Strong-SMSM. It is important to note that in the code suggested, the alphabet size does not increase with the network parameters due to the strong-security constraint.

Under an individual secrecy constraint, the converse below gives a stronger result than the min-cut bound.
\begin{theorem}\label{converse}
Assume an SMSM network $(\mathcal{V},\mathcal{E},\mathcal{S},\mathcal{D},w)$.
Under individual security for $k-w$ messages, that is, requiring $I(\textbf{M}^{k-w}_{s};\textbf{Z}_w)=0$ for any set of $k-w$ messages, one must have
\[
H(\textbf{M}_s) \leq \rho(s,d_i) - \rho(s,z) + w.
\]
\end{theorem}
This result should be interpreted as follows. If Eve observes $w$ independent links, and $\rho(s;z) = w$, then one must have $H(\textbf{M}_s) \leq \rho(s,d_i)$, which is the cut set bound. Of course, as mentioned before, the surprising part is that this bound is tight, hence such a level of security is available without any loss in rate. Yet, \Cref{converse} is slightly stronger, in the sense that if somehow Eve observes more then $w$ links, yet one still wishes to be secure with respect to any set of $k-w$ messages, then $H(\textbf{M}_s)$ should be \emph{strictly} smaller than $\rho(s,d_i)$ and by the same amount.
E.g., if Eve observes $w+e$ links, we have $H(\textbf{M}_s) \leq \rho(s,d_i)-e$. This means a linear increase in Eve's power results in a linear decrease in rate.

The achievability (direct) and the leakage proof are given in \Cref{SecureGossipAlgorithmMultiplex} using a random, non-linear code, and in \Cref{LinearCodes} using a structured linear code.
We note that using the non-linear code, the field size is determined only by the network coding scheme and its multicast structure (we elaborate about it in \Cref{Alphabet}), and there is no increase in the field size due to the security constraint. Of coarse, it requires $k$ to be large, but the code is over a binary field. On the down side, in the non-linear code, both the sources and the destinations must store a big codebook, which includes all the possible bins and codewords. We also note that the non-linear code is more straightforward, easy to understand and uses a simple binning scheme, which is common in information theory, yet is used here with a security twist. On the other hand, using the linear code suggested, there is no requirement for $k$ to grow, and the encoding/decoding is done by a linear function, such that it is not required to store a big codebook. Yet, to obtain the secrecy constraint, the code requires a field size greater than or equal to $q^{u}$ at the sources and destinations, where $u=c/\log_2(q)\geq k$. This leads to calculations over a large finite field, which are complex. In both cases, the sophisticated coding is only at the sources and destinations. The network code field size and the coding at intermediate nodes can remain small.

Finally, it is important to note that the constraint on how many messages Eve catches is set on the entire network, thus, Eve may catch $w$ messages of a single source, or $w$ messages from several sources all together. Secrecy is maintained in any case, as under individual secrecy, messages from other sources can only increase secrecy, and any network code cannot create linear combinations with other messages which reduce the secrecy level. This is another benefit of the model, and hence the network code can be any LNC, without an increase in alphabet size.
In that context, the codes given in \cite{chensecure,silva2009universal} for single source multi-cast, if considered under this individual secrecy setting, can also only increase secrecy when applied to multi-source multi-cast. Thus, such codes constitute an achievable scheme as well.
\subsection{Algebraic Gossip}
As mentioned earlier, the suggested code easily applies to algebraic gossip as well, since this can be viewed as linear network coding over a time-extended graph. The following result captures the number of rounds required to (individually) securely disseminate $k$ messages from each of the $|S|$ sources to all nodes in the network.
\begin{theorem}\label{direct theorem3}
Assume an oblivious network that floods in time $T$ with throughput $\alpha$.
Then, for $|S|$ nodes in the network with $k$ messages each, algebraic gossip spreads the $k|S|$ messages to all nodes with probability $1 - \epsilon$ after
\begin{equation*}
    T^{\prime} = T + \frac{1}{\alpha}(k|S|+\log \epsilon^{-1})
\end{equation*}
rounds, while keeping any eavesdropper which observes at most $w$ packets, ignorant with respect to any set of $k_s$ messages individually.
\end{theorem}
The proof is based on applying \Cref{direct theorem1} above, together with known results from the Gossip literature. The complete details are deferred to \Cref{Algebraic Gossip}.
Note that the result above is constant-optimal, as $T$ is the number of rounds required for a single message, hence one cannot expect less that $T^{\prime}$ above for $k|S|$ messages. This is a perfect pipelining property \cite{cohen2015network}, thus, surprisingly, one can gossip securely messages to all parties in the network, without any loss in rate and without any centralized mechanism for routing, key exchange or any other encryption mechanism, as long as the eavesdropper is interested in single messages.
\subsection{Alphabet Size}\label{Alphabet}
Without secrecy constraints, Jaggi \emph{et al}. proved that a field with size greater than or equal to the number of destinations is sufficient for multicast under LNC \cite{jaggi2005polynomial}.
However, this may not hold if it is required to keep an eavesdropper ignorant.
Cai \emph{et al}. \cite{cai2011secure} devised a code which requires a field of exponential size to obtain secrecy.
There, the field size must be larger than $\binom{|E|}{w}$.
Feldman \emph{et al}. \cite{feldman2004secure} showed that there exist networks that require a field of size at least $\Theta(|E|^{\frac{w}{2}})$.
In \cite{el2012secure}, the authors demonstrate that secure network coding can be considered as a network generalization of the wiretap channel of type II.
When $d$ is the number of destinations in the multicast connection, a field of size $\binom{2k^3d^2}{w-1+d}$ is sufficient, which is independent of $|V|$ and $|E|$ but is still exponential in other network parameters.

In the solution we suggest herein, the field size is determined only by the network coding scheme, that is, only by the requirement for \emph{reliability}, and is not increased by the individual-security constraints. In the gossip case, for example, since $q^{-\alpha k}=2^{-(\alpha \log q)k}$, any field size greater than or equal to $2$ will suffice, and an increase in the field size has only a logarithmic effect on the throughput, meaning only a logarithmic multiplier on the number of rounds $T^{\prime}$ required.

\section{Code Construction and a Proof\\for Individual-SMSM (\Cref{direct theorem1} and \Cref{direct theorem3})}\label{SecureGossipAlgorithmMultiplex}
At each source node $s\in\{1,\ldots,|S|\}$, we map each column of the message matrix $\textbf{M}_s$ to a column of the same length.
Specifically, as depicted in \Cref{fig:MultiplexWiretapCoding}, in the code construction phase, for each possible value for the partial column $M_{s,1}(i);\ldots;M_{s,k^{\prime}}(i), 1\leq i\leq c$, of length $k^{\prime}=k-(w+k\epsilon)$ in the message matrix, we generate a bin, containing several columns of length $k$. At this point, we only mention that $\epsilon$ is such that $k\epsilon$ is an integer, as it represents a number of bits. At the end of the leakage proof, we discuss how small $\epsilon$ has to be exactly and how it affects the mutual information.
The number of columns in a bin \emph{corresponds} to $w$, the number of packets that the eavesdropper can wiretap, in a relation that will be made formal in the sequel.
Then, to map the $i$-th column of the $s$-th message matrix, we select a column from the bin corresponding to $M_{s,1}(i);\ldots;M_{s,k^{\prime}}(i)$. The specific column within that bin is chosen according to $M_{s,k^{\prime}+1}(i);\ldots;M_{s,k}(i)$.
That is, the lower part of the original column points to the bin, and the upper part of the original column serves as an \emph{index} in order to choose the right column from the bin.
This way, a new, $k\times c$ message matrix $\textbf{X}_s$ is created. This message matrix contains $k$ new messages of the same size $c$.
\ifdouble
\begin{figure}
  \centering
  \includegraphics[trim= 0.2cm 0cm 0cm 0cm,clip,scale=1]{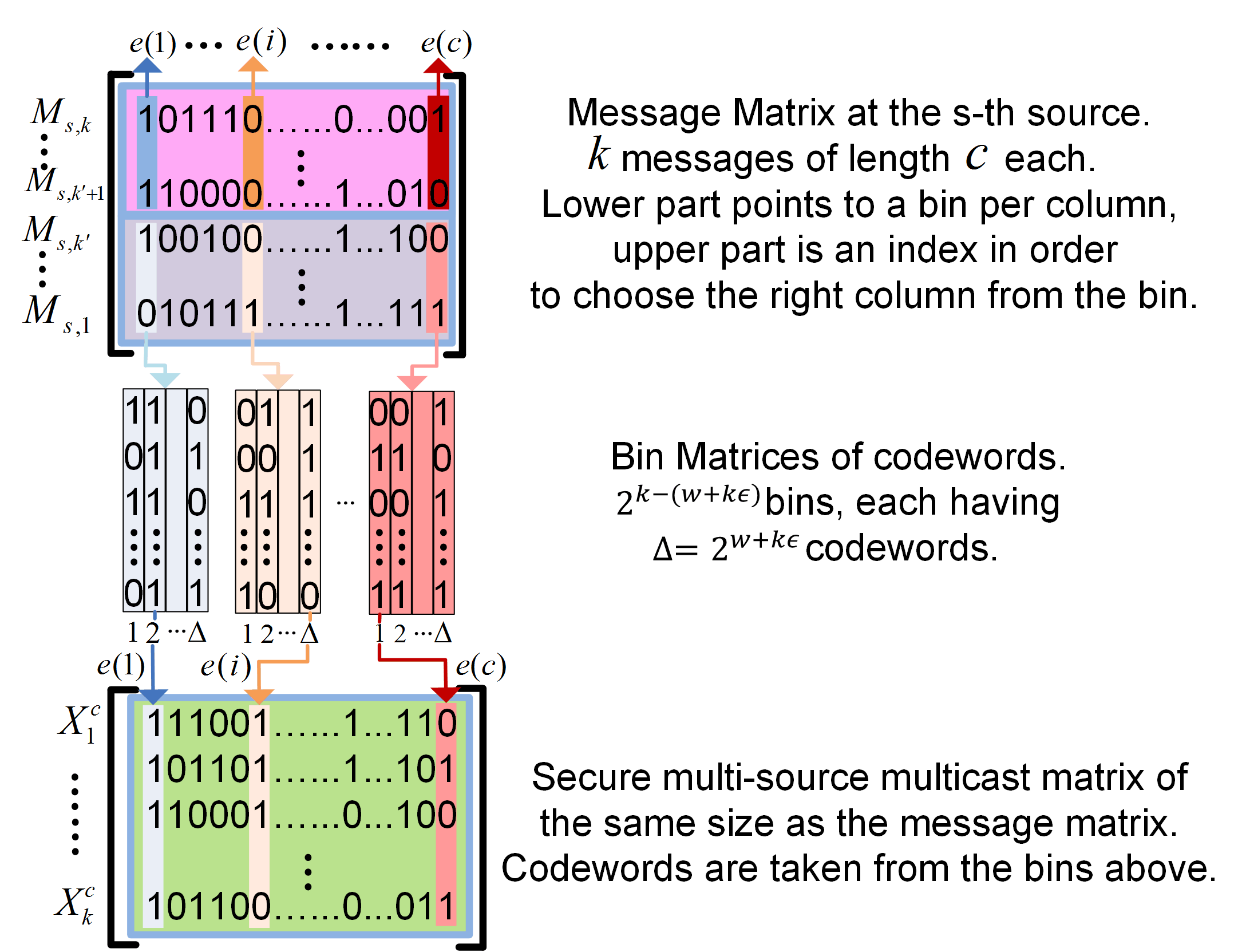}
  \caption{Binning and source encoding process for Individual-SMSM.}
  \label{fig:MultiplexWiretapCoding}
\end{figure}
\else
\begin{figure}
  \centering
  \includegraphics[trim= 0cm 0cm 0cm 0cm,clip,scale=1.25]{MultiplexWiretap_coding11_one_colM.png}
  \caption{Binning and source encoding process for Individual-SMSM.}
  \label{fig:MultiplexWiretapCoding}
\end{figure}
\fi
We may now turn to the detailed construction and analysis.
\subsubsection{Codebook Generation}
Set $\Delta = 2^{w+k\epsilon}$. Let $P(x)\sim Bernoulli(1/2)$. Using a distribution $P(X^k)=\prod^{k}_{j=1}P(x_j)$, for each possible column $M_1(i);\ldots;M_{k^{\prime}}(i)$ in the message matrix, that is, $2^{k-(w+k\epsilon)}$ possibilities, generate $\Delta$ independent and identically distributed codewords $x^{k}(e)$, $1 \leq e \leq \Delta$. Thus, we have $2^{k-(w+k\epsilon)}$ bins, each of size $2^{w+k\epsilon}$.
Note that the length of the columns in the bins is $k$, thus the codebook matrix is of the same size as $\textbf{M}_s$. The codebook is depicted in \Cref{fig:MultiplexWiretapCoding}.
\subsubsection{Source and legitimate Node encodings}
At the $s$-th source node, the encoder selects, for each column $i$ of bits $M_{s,1}(i);\ldots;M_{s,k^{\prime}}(i)$, one codeword, $x^{k}(e(i))$, from the bin indexed by $M_{s,1}(i);\ldots;M_{s,k^{\prime}}(i)$, where $e(i)=M_{s,k^{\prime+1}}(i);\ldots;M_{s,k}(i)$.
That is, $k^{\prime}=k-(w+k\epsilon)$ bits of the column choose the bin, and the remaining $w+k\epsilon$ bits choose the codeword within the bin.

Then, similar to many RLNC protocols, the sources transmit linear combinations of the rows, with random coefficients. Nodes transmit random linear combinations of the vectors in $\mathcal{S}_v$, which is maintained by each node according to the messages received at the node.
\subsection{Reliability}\label{SecureGossipAlgorithmMultiplexReliability}
The reliability proof using RLNC is almost a direct consequence of \cite{ho2006random}. Clearly, the min-cut is given by \Cref{direct theorem1}. Hence, the legitimate nodes can easily reconstruct $\textbf{X}_s$ for each $s$ (simple, non-secure, multi-source multicast). Then, each destination maps $\textbf{X}_s$ back to $\textbf{M}_s$, as per column $1\leq i\leq c$, the index of the bin in which the codeword $\textbf{X}_s$(i) resides is $M_{s,1}(i);\ldots;M_{s,k^{\prime}}(i)$ and the index of the codeword location in that bin is $M_{s,k^{\prime}+1}(i);\ldots;M_{s,k}(i)$. It is important to note that since the codebook is generated randomly, the mapping is not exactly 1:1 and there is a possibility for a repetition of codewords. However, averaged over all messages, the error probability from such a repetition is negligible, and this scheme guarantees successful decoding with high probability. Considering the number of bins and the number of codewords in each bin as given in the codebook generation phase, the analysis on the probability of successfully decoding $\textbf{M}_s(i)$ from $\textbf{X}_s(i)$ is a direct consequence using standard analysis of random coding \cite[Section 3.4]{C13}. Note also that such a repetition can also be circumvented using a random permutation of the columns rather than random binning, though analysis is more complicated due to the memory in the process.

An example, obtaining both reliability and individual secrecy for two sources, with two messages each and four legitimate destination nodes, where the eavesdropper min-cut is at most $1$, is given in Figure \ref{fig:wiretap_Secure gossip}.
Note that secure communication with respect to \emph{one message} is possible while sending two messages from each source to all destinations.

\begin{remark}
The proposed binning scheme has many similarities with the random binning scheme introduced by Wyner's seminal work on degraded wiretap channel, which relies on information theoretic principles to obtain Physical Layer security \cite{C2}. Since its publication in 1975, Wyner’s binning scheme was utilized by numerous studies, models and solutions, e.g., comprehensive surveys can be found in \cite{C13,zhou2016physical}. In the conventional binning scheme, each message is associated with a bin (i.e., the number of bins equals the number of messages) which contains multiple codewords. The number of codewords per bin depends on the capacity of the eavesdropper. Accordingly, the sent codeword incorporates both the bin index and a (private) random key which is used to select the codeword from the corresponding bin (e.g., \cite[Chapter 3]{C13}).

However, there are a few differences between the typical binning scheme and the scheme proposed herein, as well as the code construction phase. Most importantly, in the suggested scheme we utilize the binning scheme differently such that we exploit some of the messages to protect the other messages and vice versa. In particular, in the suggested scheme each column in the message matrix is partitioned into two; the first partition points to a specific bin (similar to the message itself in the conventional binning scheme), while the second partition points to a specific codeword (in contrast to the conventional binning scheme in which the codeword is chosen at random). Note that as a result, each message-matrix-column points to one possible codeword while in the usual scheme each message points to multiple codewords. Specifically, in order to achieve the full capacity of the network, the number of bins (the first partition) is equal to the total number of messages minus the number of packets that the eavesdropper can wiretap, and the number of codewords per bin (the second partition) equals the number of packets that the eavesdropper can wiretap. This, of course, means there is no private key, and part of the information is used as a key to protect the other part. The final transformation from messages to codewords is 1:1 rather than one to many.

Second, in the suggested scheme, prior to the message coding, we mix the message matrix such that each bit in the transmitted packet is associated with a mixture of bits in the original messages (note that this message mixing is done prior to and independently from the network combinations performed by the network coding mechanism). Specifically, instead of coding the message matrix rows (the messages themselves), we code the columns of the message matrix (auxiliary messages each of which is composed from one bit from each original message). Note that we still send rows of the resulting coded matrix which mean that each bit in each captured packet conveys a coded version of many bits each coming from a different message. Since messages protect one another in our scheme, this trick allows us to require only that the messages are independent, but bits within a message are not required to be independent.

Obviously, there are also several limitations to the suggested scheme compared to the vanilla-version of the Wiretap code. First, it is important to note that achieving full capacity of the network requires a compromise on the secrecy level. Specifically, the suggested scheme ensures ``Individual Secrecy", which guarantees that the eavesdropper has zero mutual information with each message individually, rather than ``conventional secrecy constraint" which requires zero leakage of information to the eavesdropper, when normalized by the message length, independently from any other message. Second, as mentioned, attaining secrecy while utilizing the message mixing technique requires that the messages be independent and uniformly distributed (a common assumption in the related literature, e.g., \cite{chensecure,silva2009universal,cai2011theory}). Third, we require the total number of messages to be sufficiently large (large $k$), as the mutual information to the eavesdropper decays as a function of the total number of messages (\Cref{L_To_Eve}). From a secrecy perspective, however, we do not need message length to grow.
\end{remark}

\subsection{Information Leakage at the Eavesdropper}\label{L_To_Eve}
We now prove the $k_s$-individual security constraint is met, that is, $I(\textbf{M}^{k_s}_{s};\textbf{Z}_w)\rightarrow 0$ as $k \rightarrow \infty$ for any set of $k_s\leq k-(w+k\epsilon)$ messages, where $\textbf{Z}_w = \textbf{W}[\textbf{X}_1, \cdots, \textbf{X}_{|S|}]$ and $\textbf{W}$ is an arbitrary encoding matrix due to the network code. At the end of this subsection, we also quantify the rate at which the mutual information can decay as a function of $k$.

In particular, for the simplest individual secrecy constraint, we wish to show that $I(M_{s,j};\textbf{Z}_{w})$ is small for all $s\in\mathcal{S}$ and all $j$. However, herein we prove a stronger result, by showing that given $\textbf{Z}_w$, Eve's information, all possibilities for any set of $k_s \leq k-(w+k\epsilon)$ messages $\textbf{M}^{k_s}_{s}$ are almost equally likely. Hence Eve has no intelligent estimation for $\textbf{M}^{k_s}_{s}$ and $M_{s,j}$.

Denote by $\mathcal{C}_k$ the random codebook and by $\textbf{X}_s$ the set of codewords corresponding to $\vec{M}_{s,1}\ldots \vec{M}_{s,k}$.
To analyze the information leakage at the eavesdropper, note that Eve has access to at most $w$ linear combinations on the rows of $\textbf{X}_s$.

Next, note that the columns of $\textbf{X}_s$ are independent (by the construction of the codebook, creating $\textbf{X}_s$ is done independently per-column; $c$ columns are used only to reduce the NC overhead). Hence, it suffices to consider the information leakage for each column $i \in \{1,\ldots,c\}$ from $\textbf{X}_s$ separately, denoted by $\textbf{X}_s$(i).

For each column $i$ of $\textbf{M}_s$, the encoder has $2^{k^\prime}$ bins, with $\Delta$ independent and identically distributed codewords in each, out of which one is selected. Hence, there is an exponential number of codewords, from the eavesdropper's perspective, that can generate a column in $\textbf{X}_s$, and we require that Eve is still confused even given the $w$ linear combinations on each column.
Let $\textbf{Z}_w(i)$ be the $w$ linear combinations Eve has on column $i$.

Hence, when the number of codewords is $2^k$,  given the $w$ linear combinations from each column in $\textbf{Z}_w(i)$, the eavesdropper has at most $2^k(1/2)^w=2^{(k-w)}$ possible codewords.
Denote $l=k-w$. Similar to the technique used in \cite{dey2015sufficiently} to prove that \emph{myopic adversaries}
are blind, we define by the shell $\mathcal{S}h (\textbf{Z}_w(i),l)$, the set of all $k$-tuples consistent with $\textbf{Z}_w(i)$. Clearly, there are $2^l$ tuples in $\mathcal{S}h (\textbf{Z}_w(i),l)$.
See \Cref{figure:StrongSecrecy} for a graphical illustration.
\ifdouble
\begin{figure}
  \centering
  \includegraphics[trim=0cm 0cm 0cm 0cm,clip,scale=1.2]{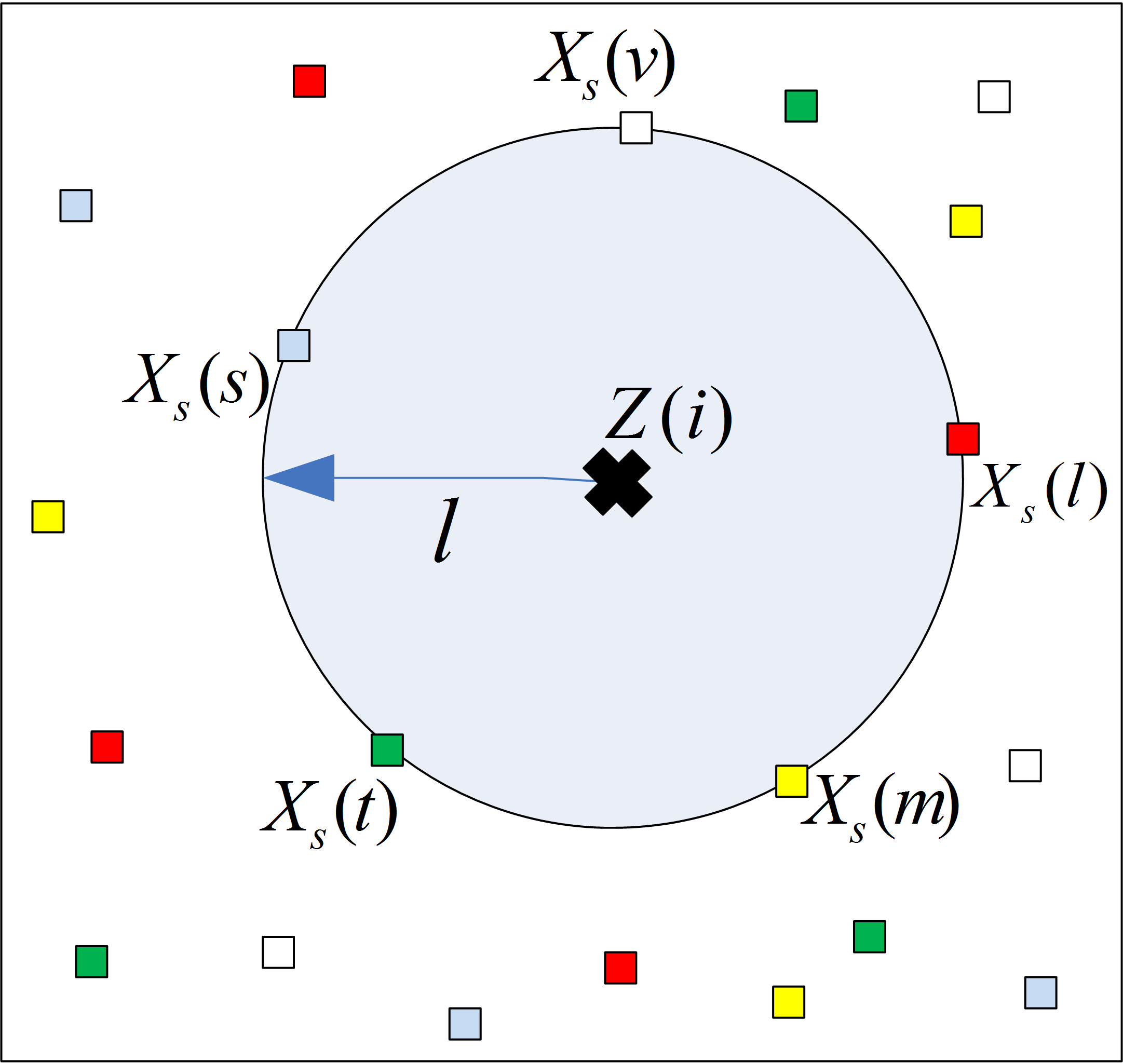}
  \caption{Codewords for Individual-SMSM algorithm lie exactly in a ball  of radius $l = k - w$ around $Z$.}
  \label{figure:StrongSecrecy}
\end{figure}
\else
\begin{figure}
  \centering
  \includegraphics[trim=0cm 0cm 0cm 0cm,clip,scale=1.5]{StrongSecrecy4.png}
  \caption{Codewords for Individual-SMSM algorithm lie exactly in a ball  of radius $l = k - w$ around $Z$.}
  \label{figure:StrongSecrecy}
\end{figure}
\fi
We assume Eve has the codebook, yet does not know which column from each bin is selected to be the codeword. Hence, we wish to show that given $\textbf{Z}_w(i)$, Eve will have at least one candidate per bin. The probability for a codeword to fall in a given shell is
\ifdouble
\begin{multline*}
Pr(\textbf{X}_{s}^{k}(i) \in \mathcal{C}_k \cap \textbf{X}_{s}^{k}(i) \in \mathcal{S}h(\textbf{Z}_w(i),l))\\
= \frac{\mathcal{V}ol(\mathcal{S}h(\textbf{Z}_w(i),l))}{2^k} =  \frac{2^{\left(k-w\right)}}{2^{k}}.
\end{multline*}
\else
\[
Pr(\textbf{X}_{s}^{k}(i) \in \mathcal{C}_k \cap \textbf{X}_{s}^{k}(i) \in \mathcal{S}h(\textbf{Z}_w(i),l))
= \frac{\mathcal{V}ol(\mathcal{S}h(\textbf{Z}_w(i),l))}{2^k} =  \frac{2^{\left(k-w\right)}}{2^{k}}.
\]
\fi
In each bin of $\mathcal{C}_k$, we have $\Delta =2^{w+k\epsilon}$ codewords. Thus, the \emph{expected} number of codewords Eve sees on a shell, \emph{per bin} is
\begin{equation*}
    \textbf{E}\left[|\{m(i):X^k(i)\in \mathcal{S}h(Z(i),l)\}|\right]  = \frac{2^{w+k\epsilon}*2^{k-w}}{2^k} = 2^{k\epsilon}.
\end{equation*}
Hence, we can conclude that on average, and if $k\epsilon$ is not too small, for every column in $\textbf{M}_s$ Eve has a few possibilities \emph{in each bin}, hence cannot locate the right bin. However, it is still important to show that all bins have (asymptotically) equally likely number of candidate codewords, hence Eve cannot locate a preferred bin. In other words, we proved that the average number of candidate codewords per column is $2^{k\epsilon}$. We now wish to show that the \emph{actual} number Eve has in each bin is concentrated around this average.

To this end, we show that now the probability that the actual number of options deviates from the average by more than $\varepsilon$ is small. Define the event
\ifdouble
\begin{multline*}
 \mathcal{E}_{C_1}(Z(i),l):= \{(1-\varepsilon^{\prime})2^{k\epsilon}\leq  \\
 |m(i):X_{s}^{k}(i)\in \mathcal{S}h(\textbf{Z}_w(i),l)| \leq (1+\varepsilon^{\prime})2^{k\epsilon}\}.
\end{multline*}
\else
\[
 \mathcal{E}_{C_1}(Z(i),l):= \{(1-\varepsilon^{\prime})2^{k\epsilon}\leq
 |m(i):X_{s}^{k}(i)\in \mathcal{S}h(\textbf{Z}_w(i),l)| \leq (1+\varepsilon^{\prime})2^{k\epsilon}\}.
\]
\fi
By the Chernoff bound, we have
\begin{equation*}
   Pr(\mathcal{E}_{C_1}(Z(i),l)) \geq 1- 2^{-\varepsilon^{\prime}2^{k\epsilon}}.
\end{equation*}

Finally, we are now able to show that the mutual information is indeed negligible. Denote by $\textbf{1}_{\mathcal{E}_{C_1}}$ the indicator for the event where the actual number of options per column does not deviates from the average. We thus have
\ifdouble
\begin{align*}
& I(\textbf{M}^{k_s}_{s}(i);\textbf{Z}_w(i)) \\
&= H(\textbf{M}^{k_s}_{s}(i))-H(\textbf{M}^{k_s}_{s}(i)|\textbf{Z}_w(i)) \\
&\leq k_s-H(\textbf{M}^{k_s}_{s}(i)|\textbf{Z}_w(i))\\
&\leq k_s-H(\textbf{M}^{k_s}_{s}(i)|\textbf{Z}_w(i),\textbf{1}_{\mathcal{E}_{C_1}})\\
&= k_s -[P(\mathcal{E}_{C_1}^c)H(\textbf{M}^{k_s}_{s}(i)|\textbf{Z}_w(i),\textbf{1}_{\mathcal{E}_{C_1}} = 0)\\
&\hspace{2.5cm} +P(\mathcal{E}_{C_1})H(\textbf{M}^{k_s}_{s}(i)|\textbf{Z}_w(i),\textbf{1}_{\mathcal{E}_{C_1}} = 1)]\\
&\leq k_s -(1- 2^{-\varepsilon^{\prime}2^{k\epsilon}})H(\textbf{M}^{k_s}_{s}(i)|\textbf{Z}_w(i),\textbf{1}_{\mathcal{E}_{C_1}} = 1)\\
&\leq k_s - H(\textbf{M}^{k_s}_{s}(i)|\textbf{Z}_w(i),\textbf{1}_{\mathcal{E}_{C_1}} = 1) + 2^{-\varepsilon^{\prime}2^{k\epsilon}}k_s.
\end{align*}
\else
\begin{align*}
& I(\textbf{M}^{k_s}_{s}(i);\textbf{Z}_w(i)) \\
&= H(\textbf{M}^{k_s}_{s}(i))-H(\textbf{M}^{k_s}_{s}(i)|\textbf{Z}_w(i)) \\
&\leq k_s-H(\textbf{M}^{k_s}_{s}(i)|\textbf{Z}_w(i))\\
&\leq k_s-H(\textbf{M}^{k_s}_{s}(i)|\textbf{Z}_w(i),\textbf{1}_{\mathcal{E}_{C_1}})\\
&= k_s -[P(\mathcal{E}_{C_1}^c)H(\textbf{M}^{k_s}_{s}(i)|\textbf{Z}_w(i),\textbf{1}_{\mathcal{E}_{C_1}} = 0) +P(\mathcal{E}_{C_1})H(\textbf{M}^{k_s}_{s}(i)|\textbf{Z}_w(i),\textbf{1}_{\mathcal{E}_{C_1}} = 1)]\\
&\leq k_s -(1- 2^{-\varepsilon^{\prime}2^{k\epsilon}})H(\textbf{M}^{k_s}_{s}(i)|\textbf{Z}_w(i),\textbf{1}_{\mathcal{E}_{C_1}} = 1)\\
&\leq k_s - H(\textbf{M}^{k_s}_{s}(i)|\textbf{Z}_w(i),\textbf{1}_{\mathcal{E}_{C_1}} = 1) + 2^{-\varepsilon^{\prime}2^{k\epsilon}}k_s.
\end{align*}
\fi
However, since we now condition on $\textbf{1}_{\mathcal{E}_{C_1}} = 1$, that is, \emph{there was no deviation from the average by more than $\varepsilon^{\prime}$}, Eve's distribution on the bins is almost uniform, and we have $H(\textbf{M}^{k_s}_{s}(i)|\textbf{Z}_w(i),\textbf{1}_{\mathcal{E}_{C_1}} = 1) \ge k_s(1- \varepsilon^{\prime})$. Thus,
\[
I(\textbf{M}^{k_s}_{s}(i);\textbf{Z}_w(i)) \leq k_s (\varepsilon^{\prime} +2^{-\varepsilon^{\prime}2^{k\epsilon}}).
\]
Note that we are free to choose small $\varepsilon^\prime$ and $\epsilon$ as long as $k\epsilon$ is an integer. E.g., taking $\varepsilon^\prime = k^{-m}$ and $\epsilon$ such that $k\epsilon = \left \lceil{(m+1)\log k}\right \rceil$, for some $m \ge 2$, gives $I(\textbf{M}^{k_s}_{s}(i);\textbf{Z}_w(i)) = O(k^{-m+1})$.
\subsection{Algebraic Gossip (\Cref{direct theorem3})}\label{Algebraic Gossip}
Revisiting the above proof of \Cref{direct theorem1}, it is clear that the algorithm suggested therein is a source code, followed by any good NC scheme which can be used to disseminate the encoded packages over the network. This is actually a \emph{universal code} \cite{silva2011universal,matsumoto2017universal}, in the sense that it can be applied to any network without requiring knowledge of, or any modifications on, the NC. Special coding is required only at the sources. Hence, using the same source code as in \Cref{direct theorem1} and then disseminating the encoded messages by a gossip protocol as in \cite{haeupler2011analyzing}, which uses RLNC, the reliability at the legitimate nodes is a direct consequence of \cite[Theorem 1]{cohen2015network}. Moreover, the mixing that the gossip protocol does over the network is the \emph{same as RLNC}, hence, the information leakage at the eavesdropper is as proved in \Cref{L_To_Eve}. Thus, compared to only a reliability constraint, the number of rounds required for both reliability and individual-secrecy is exactly the same as in the original non-secure gossip protocol.

\section{Converse (\Cref{converse})}\label{converse_sec}
In this section, we derive a converse result, which shows that under individual secrecy on a group of $k-w$ messages, not only the rate is bounded by the min-cut, but, more importantly, any independent link that Eve observes above $w$ will require to reduce the rate \emph{at the same amount} in order to achieve both reliability and secrecy.
Thus, the converse result derived herein will be specific for the ``individual secrecy" constraint given in \Cref{constraints}, and its extension to any set of $k-w$ messages, and hence extend on the well known cut-set bound.

Let $\bar{\textbf{Z}}$ denote the random variable corresponding to the links which are not available to Eve. Hence, $\textbf{Y}_d = (\textbf{Z},\bar{\textbf{Z}})$. Let $\textbf{M}_{s}^{k-w}$ denote a set of $k-w$ messages, and $\textbf{M}_{s}^w$ denote the remaining $w$. We will show that reliability, that is $H(\textbf{M}_s|\textbf{Y}_d) = 0$, and individual secrecy, that is, $I(\textbf{M}_{s}^{k-w};\textbf{Z}) = 0$, imply that $H(\textbf{M}_{s})$ is upper bounded by the term in \Cref{converse}.
\ifdouble
\begin{align}\label{eq:R_s}
&H(\textbf{M}_{s})
\nonumber\\
&= H(\textbf{M}_{s}^{k-w}|\textbf{M}_{s}^w) +H(\textbf{M}_{s}^w)
\nonumber\\
&\stackrel{(a)}{\leq} I(\textbf{M}_{s}^{k-w};\textbf{Y}_d|\textbf{M}_{s}^w) + H(\textbf{M}_{s}^{k-w}|\textbf{Y}_d) + w
\nonumber\\
&\stackrel{(b)}{=} I(\textbf{M}_{s}^{k-w};\textbf{Z},\bar{\textbf{Z}}|\textbf{M}_{s}^w)+w
\nonumber\\
&= I(\textbf{M}_{s}^{k-w};\textbf{Z}|\textbf{M}_{s}^w)+ I(\textbf{M}_{s}^{k-w};\bar{\textbf{Z}}|\textbf{Z},\textbf{M}_{s}^w)+w
\nonumber\\
&= I(\textbf{M}_{s}^{k-w};\textbf{Z})+ I(\textbf{M}_{s}^w;\textbf{Z}|\textbf{M}_{s}^{k-w}) - I(\textbf{Z};\textbf{M}_{s}^w)
\nonumber\\
& \hspace{1.2cm} + I(\textbf{M}_{s}^{k-w};\bar{\textbf{Z}}|\textbf{Z},\textbf{M}_{s}^w) +w
\nonumber\\
&\stackrel{(c)}{=} I(\textbf{M}_{s}^w;\textbf{Z}|\textbf{M}_{s}^{k-w}) - I(\textbf{Z};\textbf{M}_{s}^w) + I(\textbf{M}_{s}^{k-w};\bar{\textbf{Z}}|\textbf{Z},\textbf{M}_{s}^w) +w
\nonumber\\
&= I(\textbf{M}_{s}^w;\textbf{Z}|\textbf{M}_{s}^{k-w}) - I(\textbf{Z};\textbf{M}_{s}^w)
\nonumber\\
&\hspace{1.2cm}+ H(\bar{\textbf{Z}}|\textbf{Z},\textbf{M}_{s}^w) - H(\bar{\textbf{Z}}|\textbf{M}_{s}^{k-w},\textbf{Z},\textbf{M}_{s}^w)+w
\nonumber\\
&= I(\textbf{M}_{s}^w;\textbf{Z}|\textbf{M}_{s}^{k-w}) - I(\textbf{Z};\textbf{M}_{s}^w) + H(\bar{\textbf{Z}}|\textbf{Z},\textbf{M}_{s}^w) +w
\nonumber\\
&= H(\textbf{M}_{s}^w|\textbf{M}_{s}^{k-w})- H(\textbf{M}_{s}^w|\textbf{Z}, \textbf{M}_{s}^{k-w})-  H(\textbf{M}_{s}^w)
\nonumber\\
&\hspace{1.2cm} + H(\textbf{M}_{s}^w|\textbf{Z}) + H(\bar{\textbf{Z}}|\textbf{Z},\textbf{M}_{s}^w) + H(\bar{\textbf{Z}})- H(\bar{\textbf{Z}})+w
\nonumber\\
&=I(\textbf{M}_{s}^w ;\textbf{M}_{s}^{k-w}|\textbf{Z}) -I(\bar{\textbf{Z}}; \textbf{Z},\textbf{M}_{s}^w)+ H(\bar{\textbf{Z}})+w
\nonumber
\end{align}
\begin{align}
&= I(\textbf{M}_{s}^w ;\textbf{M}_{s}^{k-w}|\textbf{Z}) -I(\bar{\textbf{Z}}; \textbf{M}_{s}^w |\textbf{Z}) - I(\bar{\textbf{Z}}; \textbf{Z})+ H(\bar{\textbf{Z}})+w
\nonumber\\
& \leq  I(\textbf{M}_{s}^w ;\textbf{M}_{s}^{k-w}|\textbf{Z}) -I(\bar{\textbf{Z}}; \textbf{M}_{s}^w |\textbf{Z}) + H(\bar{\textbf{Z}})+w
\nonumber\\
& = H(\textbf{M}_{s}^w|\textbf{Z}) -H(\textbf{M}_{s}^w | \textbf{Z} ,\textbf{M}_{s}^{k-w}) - H(\textbf{M}_{s}^w|\textbf{Z})
\nonumber\\
& \hspace{1.2cm}+H(\textbf{M}_{s}^w | \textbf{Z} ,\bar{\textbf{Z}})+ H(\bar{\textbf{Z}})+w
\nonumber\\
&\stackrel{(d)}{\leq} H(\bar{\textbf{Z}})+w
\nonumber\\
& \stackrel{(e)}{\leq} \rho(s_i;d_i)-\rho(s_i;z)+w, \nonumber
\end{align}
\else
\begin{align}\label{eq:R_s}
&H(\textbf{M}_{s})
\nonumber\\
&= H(\textbf{M}_{s}^{k-w}|\textbf{M}_{s}^w) +H(\textbf{M}_{s}^w)
\nonumber\\
&\stackrel{(a)}{\leq} I(\textbf{M}_{s}^{k-w};\textbf{Y}_d|\textbf{M}_{s}^w) + H(\textbf{M}_{s}^{k-w}|\textbf{Y}_d) + w
\nonumber\\
&\stackrel{(b)}{=} I(\textbf{M}_{s}^{k-w};\textbf{Z},\bar{\textbf{Z}}|\textbf{M}_{s}^w)+w
\nonumber\\
&= I(\textbf{M}_{s}^{k-w};\textbf{Z}|\textbf{M}_{s}^w)+ I(\textbf{M}_{s}^{k-w};\bar{\textbf{Z}}|\textbf{Z},\textbf{M}_{s}^w)+w
\nonumber\\
&= I(\textbf{M}_{s}^{k-w};\textbf{Z})+ I(\textbf{M}_{s}^w;\textbf{Z}|\textbf{M}_{s}^{k-w}) - I(\textbf{Z};\textbf{M}_{s}^w) + I(\textbf{M}_{s}^{k-w};\bar{\textbf{Z}}|\textbf{Z},\textbf{M}_{s}^w) +w
\nonumber\\
&\stackrel{(c)}{=} I(\textbf{M}_{s}^w;\textbf{Z}|\textbf{M}_{s}^{k-w}) - I(\textbf{Z};\textbf{M}_{s}^w) + I(\textbf{M}_{s}^{k-w};\bar{\textbf{Z}}|\textbf{Z},\textbf{M}_{s}^w) +w
\nonumber\\
&= I(\textbf{M}_{s}^w;\textbf{Z}|\textbf{M}_{s}^{k-w}) - I(\textbf{Z};\textbf{M}_{s}^w) + H(\bar{\textbf{Z}}|\textbf{Z},\textbf{M}_{s}^w) - H(\bar{\textbf{Z}}|\textbf{M}_{s}^{k-w},\textbf{Z},\textbf{M}_{s}^w)+w
\nonumber\\
&= I(\textbf{M}_{s}^w;\textbf{Z}|\textbf{M}_{s}^{k-w}) - I(\textbf{Z};\textbf{M}_{s}^w) + H(\bar{\textbf{Z}}|\textbf{Z},\textbf{M}_{s}^w) +w
\nonumber\\
&= H(\textbf{M}_{s}^w|\textbf{M}_{s}^{k-w})- H(\textbf{M}_{s}^w|\textbf{Z}, \textbf{M}_{s}^{k-w})-  H(\textbf{M}_{s}^w) + H(\textbf{M}_{s}^w|\textbf{Z}) + H(\bar{\textbf{Z}}|\textbf{Z},\textbf{M}_{s}^w) + H(\bar{\textbf{Z}})- H(\bar{\textbf{Z}})+w
\nonumber\\
&=I(\textbf{M}_{s}^w ;\textbf{M}_{s}^{k-w}|\textbf{Z}) -I(\bar{\textbf{Z}}; \textbf{Z},\textbf{M}_{s}^w)+ H(\bar{\textbf{Z}})+w
\nonumber\\
&= I(\textbf{M}_{s}^w ;\textbf{M}_{s}^{k-w}|\textbf{Z}) -I(\bar{\textbf{Z}}; \textbf{M}_{s}^w |\textbf{Z}) - I(\bar{\textbf{Z}}; \textbf{Z})+ H(\bar{\textbf{Z}})+w
\nonumber\\
& \leq  I(\textbf{M}_{s}^w ;\textbf{M}_{s}^{k-w}|\textbf{Z}) -I(\bar{\textbf{Z}}; \textbf{M}_{s}^w |\textbf{Z}) + H(\bar{\textbf{Z}})+w
\nonumber\\
& = H(\textbf{M}_{s}^w|\textbf{Z}) -H(\textbf{M}_{s}^w | \textbf{Z} ,\textbf{M}_{s}^{k-w}) - H(\textbf{M}_{s}^w|\textbf{Z})+H(\textbf{M}_{s}^w | \textbf{Z} ,\bar{\textbf{Z}})+ H(\bar{\textbf{Z}})+w
\nonumber\\
&\stackrel{(d)}{\leq} H(\bar{\textbf{Z}})+w
\nonumber\\
& \stackrel{(e)}{\leq} \rho(s_i;d_i)-\rho(s_i;z)+w, \nonumber
\end{align}
\fi
where (a) is since conditioning reduces entropy. Note that this inequality is tight if the messages are independent. (b) is due to the reliability constraint, that is,  $H(\textbf{M}_{s}^{k-w}|\textbf{Y}_d) = 0$ since all messages, and $\textbf{M}_{s}^{k-w}$ specifically, are decodable at the destination using $\textbf{Y}_d = (\textbf{Z},\bar{\textbf{Z}})$. (c) follows since we assume that Eve is kept ignorant regarding any group of $w-k$ messages, hence $I(\textbf{M}_{s}^{k-w};\textbf{Z})=0$. That is, it uses the security constraint. (d) is since $H(\textbf{M}_{s}^w|\textbf{Z},\bar{\textbf{Z}})=0$. However, since the converse turns out to be tight, and it turns out that removing the positive term $H(\textbf{M}_{s}^w|\textbf{Z},\textbf{M}_{s}^{k-w})$ does not change much, we conclude that $H(\textbf{M}_{s}^w|\textbf{Z},\textbf{M}_{s}^{k-w})$ is negligible, meaning, \emph{given the messages Eve is interested in}, and her captured links, she is actually able to decode the rest of the messages/randomness as well. This is a returning theme in such wiretap-like coding schemes. (e) follows since $\rho(s_i;d_i)-\rho(s_i;z)$ is the maximum amount that may not be available to Eve, if she has a min-cut $\rho(s_i;z)$.
Again, we assume unit capacity links and normalize the information in a message to $"1"$ accordingly.

\section{Linear Code Construction for Individual-SMSM}\label{LinearCodes}
The code given in \Cref{SecureGossipAlgorithmMultiplex} relies on random coding, for which the individual secrecy constraint holds only asymptotically, that is, $H(M_{s,j}|\textbf{Z}_{w})/H(M_{s,j})\rightarrow 1$ as $k$ grows. In this section, we provide a structured linear code, which results in zero mutual information, not constrained by the number of messages (i.e., no requirement for $k$ to grow), without any decrease in the rate or any message ``blow-up" by extra randomness. Yet, the code requires a field size greater than or equal to $q^{u}$ at the sources and destinations, where $u=c/\log_2(q)\geq k$. Note that this means we regard packets as symbols over a finite field $F_{q^u}$, where $F_{q}^{1\times u} \cong F_{q^u}$. Hence, the code is compatible with the network code given in \Cref{SecureGossipAlgorithmMultiplex} since $F_{q^u}$ is a vector space over $F_{q}$. Furthermore, note that the increased alphabet size is only at the sources and destinations. The network code at intermediate nodes can remain small. The structured linear code provided in this section adopts the linear code given in \cite{bhattad2005weakly,silva2009universal} for single-source multicast to multi-source multicast. Even though most of the results and proofs in this subsection follow the techniques given in \cite{silva2011universal}, for completeness we provide the adapted proofs herein, adopting the scope, terminology and notations to the multi-source multicast problem considered in this study. Note, however, that in the information leakage proof, Eve's observation $\textbf{Z}_w$ may include a linear combination of packets from several sources,  rather than only one source as in \cite{silva2009universal}, hance maybe confusing Eve even further. We do not exploit this confusion here, and the proof holds even for the case where Eve may see $w$ packets unmixed with other sources, yet it creates a difference compared to single-source multicast. Moreover, if one could \emph{guarantee} Eve's observation include mixtures of packets from several sources, the field size of the code may have been decreased. The current literature includes several examples analyzing the security achieved compared to the level of mixing in the network \cite{lima2007random}.
\begin{corollary}
With a $(k,w)$ linear code over a field $F_{q^u}$, \emph{$k_s$-individual security} holds in SMSM networks over a field $F_{q}$, keeping an eavesdropper which observes $w \geq \rho(s;z)$ links ignorant with respect to any set of $k_s \leq k-w$ messages, if $u \geq k$, form each source $s\in\mathcal{S}$ to each destination $d\in\mathcal{D}$, $\rho(s,d) \geq k$, for all $d\in\mathcal{D}$, $\rho(S,d) \geq (k)|S|$ and $k$ satisfies
\begin{equation*}
  k\geq \Bigg\lceil \frac{\rho(s,d)}{\rho(s,d)-\rho(s;z)}\Bigg\rceil\geq 2.
\end{equation*}
\end{corollary}
We may now turn to the detailed construction and proof of the Individual-SMSM structured linear code.
\subsubsection{Codebook Generation}
Let $\mathcal{C}$ be a linear code over $F_{q^u}$ of length $k$ and dimension $w$, and set $k^{\prime} = k-w$. Then, let
\[
\textbf{H} = [\vec{H}_1; \vec{H}_2; \ldots; \vec{H}_{k^{\prime}}] \in F_{q^u}^{k^{\prime}\times k}
\]
be a parity check matrix for the code. This linear code defines $q^{u(k-w)}$ cosets, one of them is the code itself. We denote the cosets by $\{A_m\}$, $1\leq m \leq q^{u(k-w)}$. Note that each coset is of size $q^{uw}$. Hence, the cosets of this code correspond to the bins we used in \Cref{SecureGossipAlgorithmMultiplex}, yet over a field $F_{q^u}$, such that there are many more cosets, and each is larger.

Let $\textbf{G}$ be a generator matrix for $\mathcal{C}$. We thus denote
\[
\textbf{G} = [\vec{G}_1; \vec{G}_2; \ldots; \vec{G}_w] \in F_{q^u}^{w \times k},
\]
and we select a matrix
\[
\textbf{G}^{\star} = [\vec{G}^{\star}_{1}; \vec{G}^{\star}_{2}; \ldots; \vec{G}^{\star}_{k^{\prime}}]  \in F_{q^u}^{k^{\prime} \times k},
\]
with $k^{\prime}$ linearly independent rows from $F_{q^u}^{k}\setminus \mathcal{C}$. That is, $\textbf{G}^{\star}$ spans the null space of $\mathcal{C}$.

The linear code we consider herein is from the class of linear Maximum Rank Distance (MRD) codes \cite{gabidulin1985theory, roth1991maximum}.
The norm of a vector $\textbf{X}_{s} \in F_{q^u}^{k}$ is defined as the column rank of $\textbf{X}_{s}$ over $F_{q}$, denoted by $rank_{F_{q}}(\textbf{X}_{s})$.  The \emph{rank distance} of two vectors over $F_{q^u}^{k}$ is defined as
\[
d_R(\textbf{X}_{s}(1),\textbf{X}_{s}(2))\triangleq rank_{F_{q}}(\textbf{X}_{s}(1)-\textbf{X}_{s}(2)).
\]
The \emph{minimum rank distance} of a code $\mathcal{C}\subseteq F_{q^u}^{k}$ is defined as the minimum distance of all pairs of distinct codewords in $\mathcal{C}$. Rank metric codes adhere to a Singleton bound, that is, the size of a code is bounded by $|\mathcal{C}|\leq q^{\max\{k,u\}(\min\{k,u\}-d+1)}$. For linear codes this becomes $d \leq \min\{1,u/k\}(k-k^{\prime})+1$. We use the requirements in \cite[Theorem 7]{silva2009universal}, that is, for $u\geq k$, we use a code over $F_{q^u}$, defined by a parity check matrix $\textbf{H}\in F_{q^u}^{k^{\prime}\times k}$, which is MRD and the matrix $\textbf{P}\textbf{H}\textbf{T}$ is nonsingular for all full rank $\textbf{P}\in F_{q}^{k^{\prime}\times k^{\prime}}$ and all full rank $\textbf{T}\in F_{q}^{k\times k^{\prime}}$.
\subsubsection{Source and legitimate nodes encodings}
At each source node, $s$, the encoder selects a codeword $x^{k}(e)$  out of the $q^{uw}$ members of the coset $A_m$, where $m$ is given by the index $M_{s,1};\ldots;M_{s,k^{\prime}}$ and $e=M_{s,k^{\prime}+1};\ldots;M_{s,k}$ over a field $F_{q^u}^k$. That is, similar to \Cref{SecureGossipAlgorithmMultiplex}, $k^{\prime}=k-w$ symbols choose the coset, and the remaining $w$ symbols choose the codeword within the coset. This is equivalent to letting $\textbf{X}_s$ be a choice from the $q^{uw}$ solutions of
\begin{equation}\label{eq:pro2_1}
(M_{s,1};\ldots;M_{s,k^{\prime}}) = \textbf{H}\textbf{X}_{s}.
\end{equation}
Again, note that $\textbf{X}_s$ is of the same size as $\textbf{M}_s$. \Cref{linear encoding prop} below shows that, in fact, $\textbf{X}_s$ can be easily computed using matrix multiplication.
\begin{proposition}\label{linear encoding prop}
At each source node, $s$, the encoding operation for the symbols $M_{s,1};\ldots;M_{s,k}\triangleq \textbf{M}_{s}$, is given by
\ifdouble
\begin{eqnarray}\label{eq:pro2_2}
\textbf{X}_{s}^T &=& M_{s,1}\vec{G}^{\star}_{1} + \ldots + M_{s,k^{\prime}}\vec{G}^{\star}_{k^{\prime}} \nonumber\\
&&+ M_{s,k^{\prime}+1}\vec{G}_1 + \ldots + M_{s,k}\vec{G}_w \nonumber\\
&=& \textbf{M}_{s}^T \left[\begin{array}{c} \textbf{G}^{\star} \\ \textbf{G} \end{array}\right].
\end{eqnarray}
\else
\begin{equation}\label{eq:pro2_2}
\textbf{X}_{s}^T = M_{s,1}\vec{G}^{\star}_{1} + \ldots + M_{s,k^{\prime}}\vec{G}^{\star}_{k^{\prime}} + M_{s,k^{\prime}+1}\vec{G}_1 + \ldots + M_{s,k}\vec{G}_w = \textbf{M}_{s}^T \left[\begin{array}{c} \textbf{G}^{\star} \\ \textbf{G} \end{array}\right].
\end{equation}
\fi
\end{proposition}
\begin{proof}
Define $\textbf{X}_{s}^{T}$ according to \eqref{eq:pro2_2}. We wish to show that this definition is indeed consistent with \eqref{eq:pro2_1}, that is, using the definition in \eqref{eq:pro2_2} the symbols  $M_{s,1};\ldots;M_{s,k^{\prime}}$ define the coset in which $\textbf{X}_{s}$ resides, and, furthermore, the remaining $w$ symbols, $M_{s,k^{\prime}+1};\ldots;M_{s,k}$, uniquely define the word within the coset.

To this end, take the transposed of equation \eqref{eq:pro2_2}, and multiply it by $\textbf{H}$. We have:
\ifdouble
\begin{eqnarray*}\label{eq:pro2_3}
\textbf{H}\textbf{X}_{s} &=& \textbf{H}M_{s,1}(\vec{G}^{\star}_{1})^{T} + \ldots + \textbf{H}M_{s,k^{\prime}}(\vec{G}^{\star}_{k^{\prime}})^{T} \nonumber\\
&&+ \textbf{H}M_{s,k^{\prime}+1}(\vec{G}_{1})^{T} + \ldots + \textbf{H}M_{s,k}(\vec{G}_{w})^T \nonumber\\
&=& \textbf{M}_{s,1} \left(\begin{array}{c} 1 \\ 0 \\ \vdots \\0 \end{array}\right) + \ldots +\textbf{M}_{s,k^{\prime}} \left(\begin{array}{c} 0\\ 0 \\ \vdots \\1 \end{array}\right)\nonumber\\
&&\qquad+ M_{s,k^{\prime}+1} \underline{0} + \ldots + M_{s,k}\underline{0} \nonumber\\
&=& \left(\begin{array}{c} \textbf{M}_{s,1} \\ \vdots \\ \textbf{M}_{s,k^{\prime}}  \end{array}\right),
\end{eqnarray*}
\else
\begin{eqnarray*}\label{eq:pro2_3}
\textbf{H}\textbf{X}_{s} &=& \textbf{H}M_{s,1}(\vec{G}^{\star}_{1})^{T} + \ldots + \textbf{H}M_{s,k^{\prime}}(\vec{G}^{\star}_{k^{\prime}})^{T} + \textbf{H}M_{s,k^{\prime}+1}(\vec{G}_{1})^{T} + \ldots + \textbf{H}M_{s,k}(\vec{G}_{w})^T \nonumber\\
&=& \textbf{M}_{s,1} \left(\begin{array}{c} 1 \\ 0 \\ \vdots \\0 \end{array}\right) + \ldots +\textbf{M}_{s,k^{\prime}} \left(\begin{array}{c} 0\\ 0 \\ \vdots \\1 \end{array}\right) + M_{s,k^{\prime}+1} \underline{0} + \ldots + M_{s,k}\underline{0} \nonumber\\
&=& \left(\begin{array}{c} \textbf{M}_{s,1} \\ \vdots \\ \textbf{M}_{s,k^{\prime}}  \end{array}\right),
\end{eqnarray*}
\fi
where the second inequality is since the row vectors $\vec{G}^{\star}$ are our choice of a basis to the null space of the code, hence, we can take $\textbf{G}^{\star}$ such that $\textbf{H}\textbf{G}^{\star T}=\textbf{I}$ \cite[Section V.C]{silva2011universal}; Moreover, since $\textbf{H}$ is a parity check matrix for the code, it is orthogonal to all codewords. Thus, the first $k^{\prime}$ rows define the coset.

Now, since the rows $M_{s,k^{\prime}+1};\ldots;M_{s,k}$ create a linear combination of \emph{codewords}, the addition of such a linear combination does not change the coset. $\textbf{X}_{s}^{T}$  remains in the same coset regardless of these rows. Yet, as ($\vec{G}_1; \ldots; \vec{G}_w)$ is of rank $w$, all $q^{uw}$ possibilities for the linear combination are distinct, creating distinct vectors $\textbf{X}_{s}^{T}$ within the coset.
\end{proof}
Then, since each symbol over $F_{q^u}$ of the encoded codeword $X_{s,1};\ldots;X_{s,k}\triangleq \textbf{X}_{s}$ at the source is a vector over $F_{q}$ of length $c=u/\log_2(q)$, the network code is still compatible with \Cref{SecureGossipAlgorithmMultiplex}. That is, the sources transmit linear combinations of the rows, with random coefficients. Nodes transmit random linear combinations of the messages they received.
\subsection{Reliability}
As for the reconstruction of $\textbf{X}_{s}$ at the destinations, it is almost a direct consequence of \cite{ho2006random}. Again, the min-cut is given by \Cref{direct theorem1}, and the legitimate nodes can easily reconstruct $\textbf{X}_s$ for each $s$. Now, each destination can map $\textbf{X}_s$ back to $\textbf{M}_s$.
First, compute the bin index according to $(M_{s,1};\ldots;M_{s,k^{\prime}}) = \textbf{H}\textbf{X}_{s}$. Then, $M_{s,k^{\prime}+1};\ldots;M_{s,k}$ are simply the index of $\textbf{X}_{s}$ within that bin. They can be computed using $(M_{s,k^{\prime}+1};\ldots;M_{s,k}) = \tilde{\textbf{G}}\textbf{X}_{s}$, where $\tilde{\textbf{G}}$ is any basis for the code such that $\tilde{\textbf{G}}\textbf{G}^{\star}=0$ yet $\tilde{\textbf{G}}\textbf{G}=I$.\footnote{In the same way, using a gossip protocol, the reliability proof is almost a direct consequence of \cite[Theorem 1]{cohen2015network}. Hence, the number of rounds required is given by \Cref{direct theorem3}.}
\subsection{Information Leakage at the Eavesdropper}\label{linear-leakage}
Denoted by $\mathcal{C}$ the code and by $\textbf{X}_s$ the codeword corresponding to $\vec{M}_{s,1}\ldots \vec{M}_{s,k}$. We assume that the eavesdropper has full knowledge of the code $\mathcal{C}$. As given in \Cref{L_To_Eve}, to analyze the information leakage at the eavesdropper, note that Eve has access to at most $w$ linear combinations on the elements of $\textbf{X}_s$.

Next, using techniques given in \cite{silva2009universal,silva2011universal}, we calculate the eavesdropper's uncertainty.
Let $\underline{\textbf{M}}_{s}^{k-w}=\textbf{P}\textbf{M}_{s}^{k-w}$ for a full rank $\textbf{P}\in F_{q}^{k^{\prime}\times k^{\prime}}$. Let $\textbf{Z}_w = \textbf{W}\textbf{X}_{s}$, where $\textbf{W}$ is an arbitrary encoding matrix due to network links observed by the eavesdropper. Then, we have
\ifdouble
\begin{eqnarray*}
  I(\underline{\textbf{M}}_{s}^{k-w};\textbf{Z}_w) &=&  I(\underline{\textbf{M}}_{s}^{k-w},\textbf{X}_{s};\textbf{Z}_w)-I(\textbf{X}_{s};\textbf{Z}_w|\underline{\textbf{M}}_{s}^{k-w})\\
  &\stackrel{(a)}{\leq}& H(\textbf{Z}_w) - H(\textbf{X}_{s}|\underline{\textbf{M}}_{s}^{k-w})\\
  && \hspace{2.6cm} + H(\textbf{X}_{s}|\underline{\textbf{M}}_{s}^{k-w},\textbf{Z}_w) \\
  &\leq& rank (\textbf{W})+ rank (\textbf{PH}) - rank\left[\begin{array}{c} \textbf{PH} \\ \textbf{W} \end{array}\right],
\end{eqnarray*}
\else
\begin{eqnarray*}
  I(\underline{\textbf{M}}_{s}^{k-w};\textbf{Z}_w) &=&  I(\underline{\textbf{M}}_{s}^{k-w},\textbf{X}_{s};\textbf{Z}_w)-I(\textbf{X}_{s};\textbf{Z}_w|\underline{\textbf{M}}_{s}^{k-w})\\
  &\stackrel{(a)}{\leq}& H(\textbf{Z}_w) - H(\textbf{X}_{s}|\underline{\textbf{M}}_{s}^{k-w})+ H(\textbf{X}_{s}|\underline{\textbf{M}}_{s}^{k-w},\textbf{Z}_w) \\
  &\leq& rank (\textbf{W})+ rank (\textbf{PH}) - rank\left[\begin{array}{c} \textbf{PH} \\ \textbf{W} \end{array}\right],
\end{eqnarray*}
\fi
where (a) is since given the network coefficients and $\textbf{X}_{s}$, $H(\textbf{Z}_w|\underline{\textbf{M}}_{s}^{k-w},\textbf{X}_{s})\geq 0$, and the last inequality follows directly from the proof of Lemma 6 in \cite{silva2011universal}. Note that the above inequality, $H(\textbf{Z}_w|\underline{\textbf{M}}_{s}^{k-w},\textbf{X}_{s})\geq 0$, is a key difference compared to single source multicast. In single source multicast, $H(\textbf{Z}_w|\underline{\textbf{M}}_{s}^{k-w},\textbf{X}_{s}) = 0$, as the output at Eve's side is determined by the source input. However, in MSM this  entropy might be positive, if Eve's observation includes a mixture with other sources and not only the source $s$.

Now, if $\langle\cdot\rangle$ denotes the row space of a matrix, we have
\[
 rank (\textbf{PH}) + rank (\textbf{W}) - rank \left[\begin{array}{c} \textbf{PH} \\ \textbf{W} \end{array}\right] = dim(\langle\textbf{PH}\rangle\cap\langle\textbf{W}\rangle).
\]
Let $r =  dim(\langle\textbf{PH}\rangle\cap\langle\textbf{W}\rangle)$. Then, there exist full rank matrices $\textbf{R}_1\in F_{q^u}^{r \times w}$ and $\textbf{R}_2\in F_{q^u}^{r \times k^{\prime}}$, such that $\textbf{R}_1\textbf{W}=\textbf{R}_2\textbf{PH}$ and $rank (\textbf{R}_2\textbf{PH})=r$, and we have $\textbf{R}_1\textbf{Z}_w = \textbf{R}_1\textbf{W}\textbf{X}_{s}=\textbf{R}_2\textbf{PH}\textbf{X}_{s}=\textbf{R}_2\underline{\textbf{M}}_{s}^{k-w}$.
On the other hand, the mutual information is at least the linear common part, that is, $I(\underline{\textbf{M}}_{s}^{k-w};\textbf{Z}_w) \geq H(\textbf{R}_2\underline{\textbf{M}}_{s}^{k-w})$. By the uniformity of the messages and the full ranks of $\textbf{P}$ and $\textbf{R}_2$, $H(\textbf{R}_2\underline{\textbf{M}}_{s}^{k-w})=r$, hence
\[
I(\underline{\textbf{M}}_{s}^{k-w};\textbf{Z}_w) \geq rank (\textbf{PH}) + rank (\textbf{W}) - rank \left[\begin{array}{c} \textbf{PH} \\ \textbf{W} \end{array}\right].
\]
To conclude,
\[
I(\underline{\textbf{M}}_{s}^{k-w},\textbf{Z}_w)= rank (\textbf{PH}) + rank (\textbf{W}) - rank \left[\begin{array}{c} \textbf{PH} \\ \textbf{W} \end{array}\right].
\]
Now, if $\mathcal{C}$ is an MRD code over a field of $F_{q^u}^{k}, u\geq k$, for any full rank $\textbf{P}\in F_{q}^{k^{\prime}\times k^{\prime}}$ the matrix $\left[\begin{array}{c} \textbf{PH} \\ \textbf{W} \end{array}\right]$ is nonsingular for any full rank $\textbf{W}\in F_{q}^{w \times k}$. Thus, $rank\left[\begin{array}{c} \textbf{PH} \\ \textbf{W} \end{array}\right] =rank (\textbf{PH})+rank (\textbf{W})$, and the mutual information is zero.

\section{Code Construction and a Proof for Strong-SMSM}\label{strong-SMSM}
In this section, we design a random code, which results with strong-secrecy, i.e., requiring Eve's mutual information \emph{with all messages simultaneously} to be zero, yet, at price of rate as given in \cite{cai2011secure}.
However, using the suggested random code herein, the field size is determined only by the network coding scheme, that is, only by
the requirement for reliability, and is not increased by the strong-security constraints.

At each source node $s\in\{1,\ldots,|S|\}$, we \emph{randomly} map each column of the message matrix $\textbf{M}_s$.
As depicted in \Cref{fig:WiretapCoding}, in the code construction phase, for each \emph{possible column} of the $s$-th message matrix we generate a bin, containing several columns. The number of such columns \emph{corresponds} to $w$, the number of packets that the eavesdropper can wiretap, in a relation that will be made formal in the sequel. Then, to encode, for each column of the message matrix, we randomly select a column from its corresponding bin. This way, a new, $n\times c$ message matrix $\textbf{X}_s$ is created.
\ifdouble
\begin{figure}
  \centering
  \includegraphics[trim= 0cm 0cm 0cm 0cm,clip,scale=1]{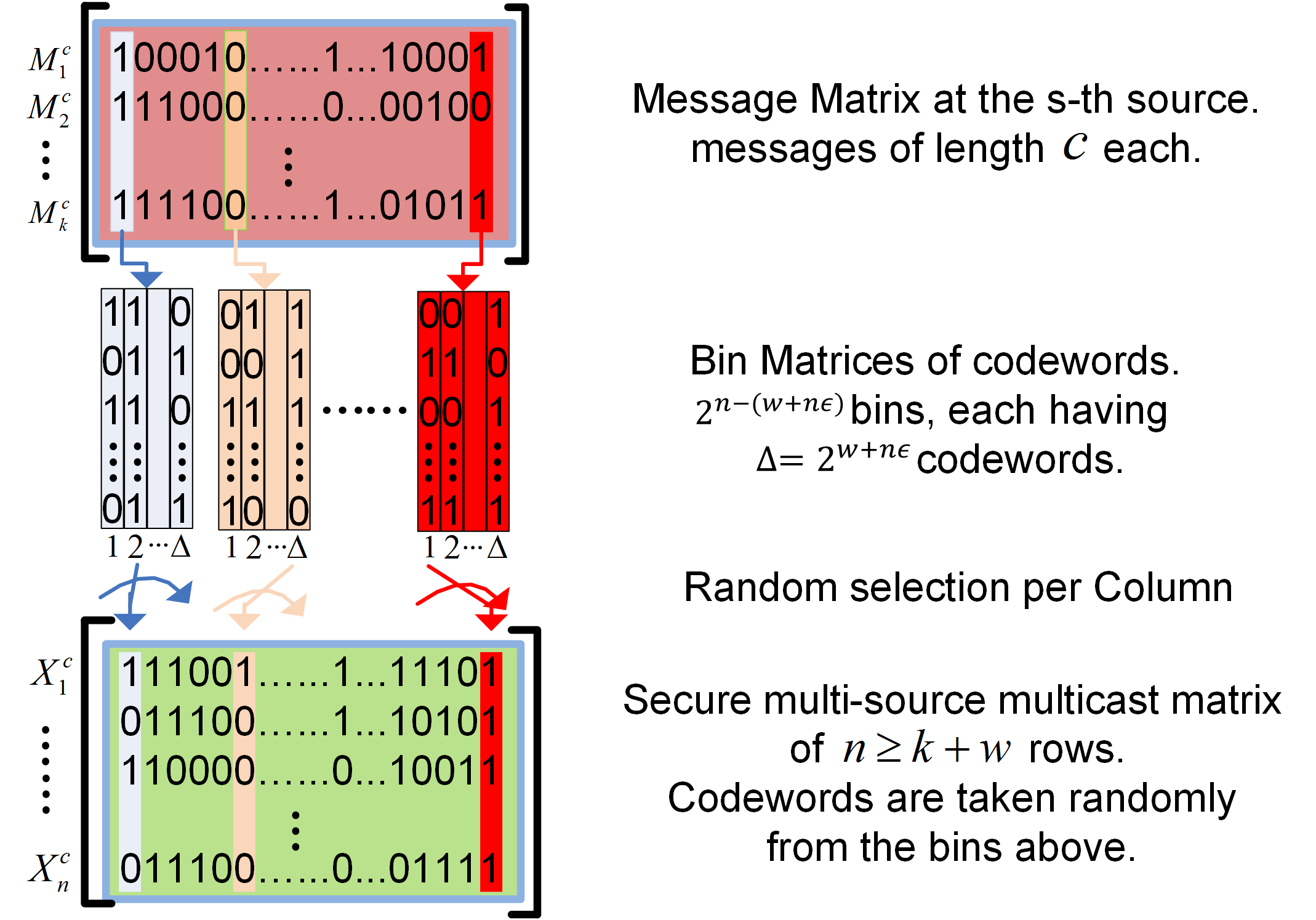}
  \caption{Binning and source encoding process for Strong-SMSM.}
  \label{fig:WiretapCoding}
\end{figure}
\else
\begin{figure}
  \centering
  \includegraphics[trim= 0cm 0cm 0cm 0cm,clip,scale=1.1]{Wiretap_coding8_one_colM.png}
  \caption{Binning and source encoding process for Strong-SMSM.}
  \label{fig:WiretapCoding}
\end{figure}
\fi
Specifically, a Strong-SMSM code at the $s$-th source node consists of a messages matrix $\textbf{M}_s$ of $\vec{M}_{s,1}\ldots \vec{M}_{s,k}$ messages of length $c$ bits over the binary field, we denote the set of matrices by $\mathcal{M}_s$; A discrete memoryless source of randomness over the alphabet $\mathcal{R}$ and some known statistics $p_R$; An encoder,
\begin{equation*}
f : \mathcal{M}_s \times \mathcal{R} \rightarrow \mathcal{X}_s\in\{0,1\}^{n\times c}
\end{equation*}
which maps each message matrix $\textbf{M}_s$ to a matrix $\textbf{X}_s$ of codewords.
This message matrix contains $n\geq k+w+n\epsilon$ new messages of size $c$, where, here as well, $n\epsilon\geq 1$ is a small integer.

The need for a \emph{stochastic encoder} is similar to most encoders ensuring information theoretic security, as randomness is required to confuse the eavesdropper about the actual information \cite{C13}. Hence, we define by $R_k$ the random variable encompassing the randomness required \emph{for the $k$ messages at the source node}, and by $\Delta$ the number of columns in each bin.
We may now turn to the detailed construction and analysis.
\subsubsection{Codebook Generation}
Set $\Delta = 2^{w+n\epsilon}$. Where $P(x)\sim Bernoulli(1/2)$, using a distribution $P(X^n)=\prod^{n}_{j=1}P(x_j)$, for each possible column in the message matrix generate $\Delta$ independent and identically distributed codewords $x^{n}(e)$, $1 \leq e \leq \Delta$.
\subsubsection{Source and legitimate node encodings}
For each column $i$ of the $s$-th message matrix $\textbf{M}_s$, the $s$-th source node selects uniformly at random one codeword $x^{n}(e)$ from the $i$-th bin.
Therefore, the $s$-th source Strong-SMSM matrix $\textbf{X}_s$ contains $c$ randomly selected codewords of length $n$, one for each column of the $s$-th message matrix.
Then, the sources transmit linear combinations of the rows, with random coefficients. Nodes transmit random linear combinations of the vectors in $\mathcal{S}_v$, which is maintained by each node according to the messages received at the node.

The reliability in the Strong-SMSM algorithm is inherited from the reliability in RLNC. That is, if min-cuts are $\rho(s,d) \geq k+w$ and $\rho(S,d) \geq (k+w)|S|$ for each $s\in\mathcal{S}$ and $d\in\mathcal{D}$ then $k+w=n$ messages can be transmitted reliably from each source to all destinations.
Since the transformation $\textbf{M}_s$ to $\textbf{X}_s$ can be inverted as given in \Cref{SecureGossipAlgorithmMultiplexReliability}, the destinations can decode the original messages.
\subsection{Information Leakage at the Eavesdropper}
We now prove the strong-security constraint is met.
In particular, for the strong constraint, we wish to show that $I(\textbf{M}_s;\textbf{Z}_{w})$ is small for all $s\in\mathcal{S}$.
We will do that by showing that given $\textbf{Z}_w = \textbf{W}[\textbf{X}_1, \cdots, \textbf{X}_{|S|}]$  where $\textbf{W}$ is arbitrary encoding matrix due to network, Eve's information, all possibilities for $\textbf{M}_s$ are equally likely, hence Eve has no intelligent estimation for $\textbf{M}_s$.

Denote by $\mathcal{C}_n$ the random codebook and by $\textbf{X}_s$ the set of codewords corresponding to $\vec{M}_{s,1}\ldots \vec{M}_{s,k}$.
To analyze the information leakage at the eavesdropper, note that Eve has access to at most $w$ linear combinations on the rows of $\textbf{X}_s$.

Next, note that the columns of $\textbf{X}_s$ are independent (by the construction of the codebook, creating $\textbf{X}_s$ is done independently per-column; $c$ columns are used only to reduce the NC overhead). Hence, it suffices to consider the information leakage for each column $i \in \{1,\ldots,c\}$ from $\textbf{X}_s$ separately.
For each column $i$ of $\textbf{M}_s$, the encoder has $\Delta$ independent and identically distributed codewords, out of which one is selected. Hence, there is an exponential number of codewords, from the eavesdropper's perspective, that can generate a column in $\textbf{X}_s$, and we require that Eve is still confused even given the $w$ linear combinations from each column.
Hence, when the number of codewords is $2^n$,  given the $w$ linear combinations from each column in $\textbf{Z}_w(i)$, the eavesdropper has $2^n(1/2)^w=2^{(n-w)}$ possible codewords. We now denote $l=n-w$ and define the shell $\mathcal{S}h (\textbf{Z}_w(i),l)$, the set of all $n$-tuples consistent with $\textbf{Z}_w(i)$. Clearly, there are $2^l$ tuples in $\mathcal{S}h (\textbf{Z}_w(i),l)$.

We assume Eve has the codebook, yet does not know which column from each bin is selected to be the codeword. Hence, we wish to show that given $\textbf{Z}_w(i)$, Eve will have at least one candidate per bin. Now,
\ifdouble
\begin{multline*}
Pr(\textbf{X}_{s}^{n}(i) \in \mathcal{C}_n \cap \textbf{X}_{s}^{n}(i) \in \mathcal{S}h(\textbf{Z}_w(i),l))\\
= \frac{\mathcal{V}ol(\mathcal{S}h(\textbf{Z}_w(i),l))}{2^n} =  \frac{2^{\left(n-w\right)}}{2^{n}}.
\end{multline*}
\else
\[
Pr(\textbf{X}_{s}^{n}(i) \in \mathcal{C}_n \cap \textbf{X}_{s}^{n}(i) \in \mathcal{S}h(\textbf{Z}_w(i),l))
= \frac{\mathcal{V}ol(\mathcal{S}h(\textbf{Z}_w(i),l))}{2^n} =  \frac{2^{\left(n-w\right)}}{2^{n}}.
\]
\fi
In each bin of $\mathcal{C}_n$, we have $\Delta =2^{w+n\epsilon}$ codewords. Thus, the \emph{expected} number of codewords Eve sees in her shell, \emph{per bin} is
\begin{equation*}
\textbf{E}\left[|\{m(i):X^n(i)\in \mathcal{S}h(Z(i),l)\}|\right]  = \frac{2^{w+n\epsilon}*2^{n-w}}{2^n} = 2^{n\epsilon}.
\end{equation*}
Again, we can conclude that on average, and if $n\epsilon$ is not too small, for every column in $\textbf{M}_s$ Eve has a few possibilities \emph{in each bin}, hence cannot locate the right bin. We need to show that all bins have (asymptotically) equally likely number of candidate codewords. Similarly to the individual security proof, we wish to show that the probability that the actual number of options deviates from the average by more than $\varepsilon$ is small. We define $\mathcal{E}_{C_1}(Z(i),l)$ similarly to \Cref{L_To_Eve} and by the Chernoff bound, we have
\begin{equation*}
  Pr(\mathcal{E}_{C_1}(Z(i),l)) \geq 1- 2^{-\varepsilon^{\prime}2^{n\epsilon}}.
\end{equation*}
The reminder of the leakage proof follows the exact same steps as the one in \Cref{L_To_Eve}, yet with $\textbf{M}_s(i)$ replacing $\textbf{M}^{k_s}_s(i)$ and $k$ replacing $k_s$.

\section{Applications}\label{applications}
\ifdouble
\begin{figure}
\centering
\includegraphics[trim= 0cm 0cm 0cm 0cm,clip,scale=0.8]{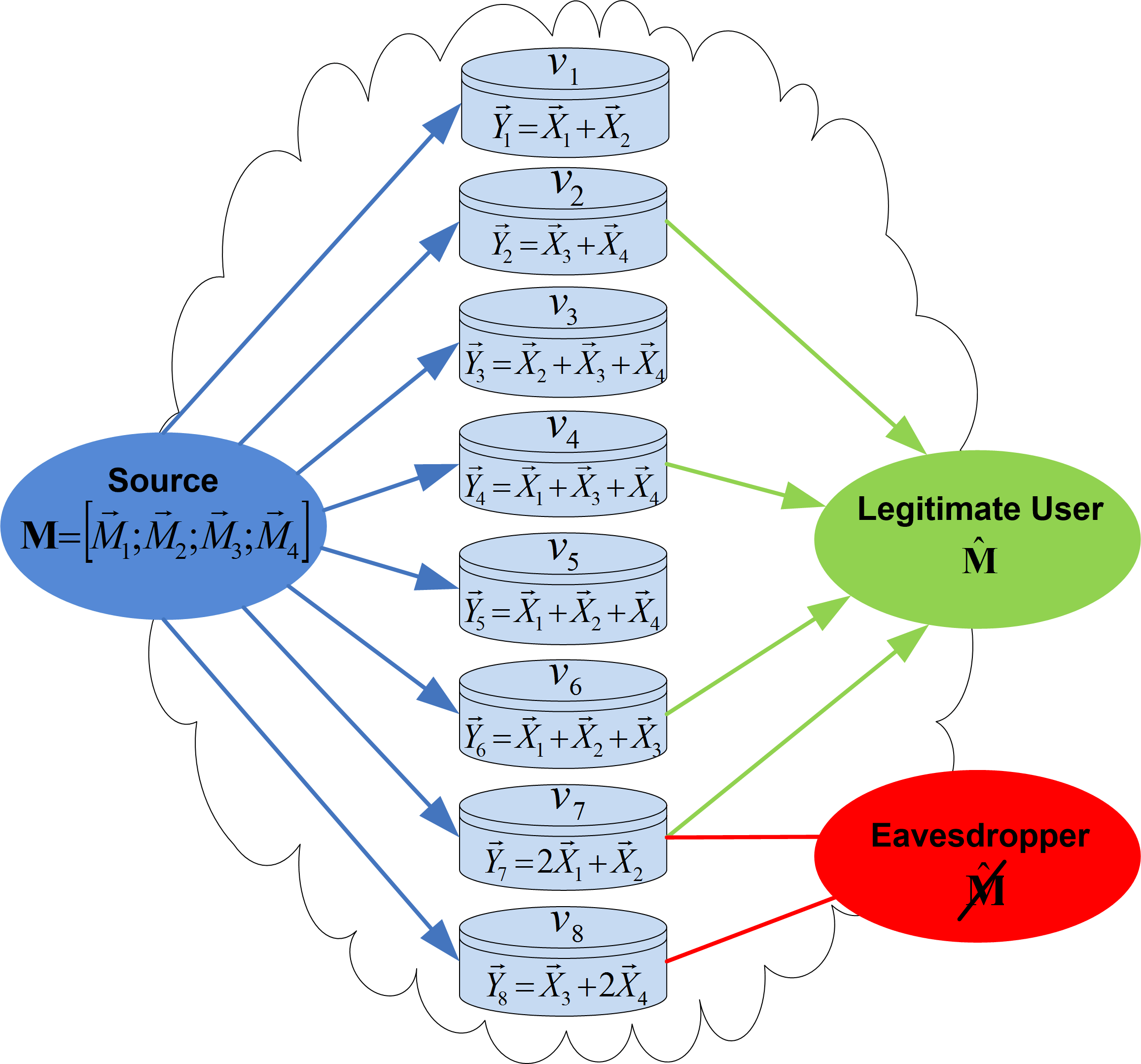}
\caption{Individual Secure Data Center, with $8$ servers. The source needs to store a file $\textbf{M}$ with $4$ messages, where any legitimate user (destination) which is connected to $4$ servers should be able to decode the $4$ original messages. In the individual secure coding scheme of this application, we assume the existence of an eavesdropper, which is able to obtain information from any $2$ servers. We wish that this eavesdropper will not gain any information on each specific message.}
\label{figure:DSS_model}
\end{figure}
\else
\begin{figure}[h]
\centering
\includegraphics[trim= 0cm 0cm 0cm 0cm,clip,scale=0.8]{DSS_4.png}
\caption{Individual Secure Data Center, with $8$ servers. The source needs to store a file $\textbf{M}$ with $4$ messages, where any legitimate user (destination) which is connected to $4$ servers should be able to decode the $4$ original messages. In the individual secure coding scheme of this application, we assume the existence of an eavesdropper, which is able to obtain information from any $2$ servers. We wish that this eavesdropper will not gain any information on each specific message.}
\label{figure:DSS_model}
\end{figure}
\fi
In previous sections we suggested an SMSM code and proved that under the suggested code an eavesdropper which can capture a subset of the packet's traversing the network (up to $w$ packets) is kept ignorant regarding each packet’s content, under the \emph{Individual Security} constraint, without compromising the rate (i.e., achieving full network capacity). In this section, we show several common applications which exemplify the applicability of the suggested code to a diverse range of protocols and applications. The first two examples include only a single source, merely to show the applicability of the \emph{individual secrecy} setup. The third example is multi-source in nature, and includes all aspects of our solution.
\subsection{Data Centers}
One of the most prominent facilities characterizing our new “information explosion” era are distributed \emph{Data Centers}. Such facilities, which aim to cope with the rapidly increasing volumes of data generated, archived and expected to be accessible, are vital to many services such as video sharing, social networks, peer-to-peer cloud storage and many more. Google's GFS \cite{ghemawat2003google}, Amazon's Dynamo \cite{decandia2007dynamo}, Google's BigTable \cite{chang2008bigtable}, Facebook's Apache Hadoop \cite{borthakur2011apache}, Microsoft's WAS \cite{calder2011windows} and LinkedIn's Voldemort \cite{auradkar2012data} are just a few examples of such ubiquitous applications. Obviously, the security and reliability of such \emph{Data Centers} are critical for such applications to be adopted by users and organizations.

In the basic non-secure model \cite{dimakis2010network, dimakis2011survey}, a source $s$ needs to store a file $\textbf{M}$, which is decomposed into $k$ messages, in $v$ servers (nodes), such that any legitimate user $d$ (destination) can reconstruct the file by collecting the stored information from any $l$ servers $(l = \rho(s,d_i)\geq k)$. With one source, as considered in \cite{kadhe2014weaklyNetCod,kadhe2014weakly,paunkoska2016improved}, the secured version constraints the stored chunks such that an eavesdropper, which can observe the information stored at any $w$ servers, will be kept ignorant regarding the actual file stored (see \Cref{figure:DSS_model}). In these works, which consider only one source, a source code to obtain weak secrecy is considered as an outer code, with a loss of a small factor of storage secrecy capacity. Then, \emph{Regenerating Codes} \cite{dimakis2010network, dimakis2011survey,rashmi2011optimal} are used. These are usually suggested to store data in distributed storage.

For the secured multi-source version, we can leverage the individual-SMSM coding scheme suggested herein to enhance the non-secure solution suggested in \cite{acedanski2005good,haeupler2011one,deb2006algebraic,haeupler2011optimality,fitzek2014implementation}, which consider each node in the network as a server which maintains pieces of data using RLNC. We will be able to guarantee that any eavesdropper that can access any $w$ servers will have no information regarding any stored message individually (zero mutual information regarding each message separately).  Specifically, each source $s$ encodes the original data file $\textbf{M}$ using the individual security coding scheme suggested herein (\Cref{SecureGossipAlgorithmMultiplex,LinearCodes}) and then uploads the encoded packets to the $v$ servers. The number and the size of packets uploaded to the servers in the secure solution suggested are as in the non-secure model; thus, we obtain the full capacity of the system.
\begin{figure}
\centering
\includegraphics[trim= 0cm 0cm 0cm 0cm,clip,scale=0.5]{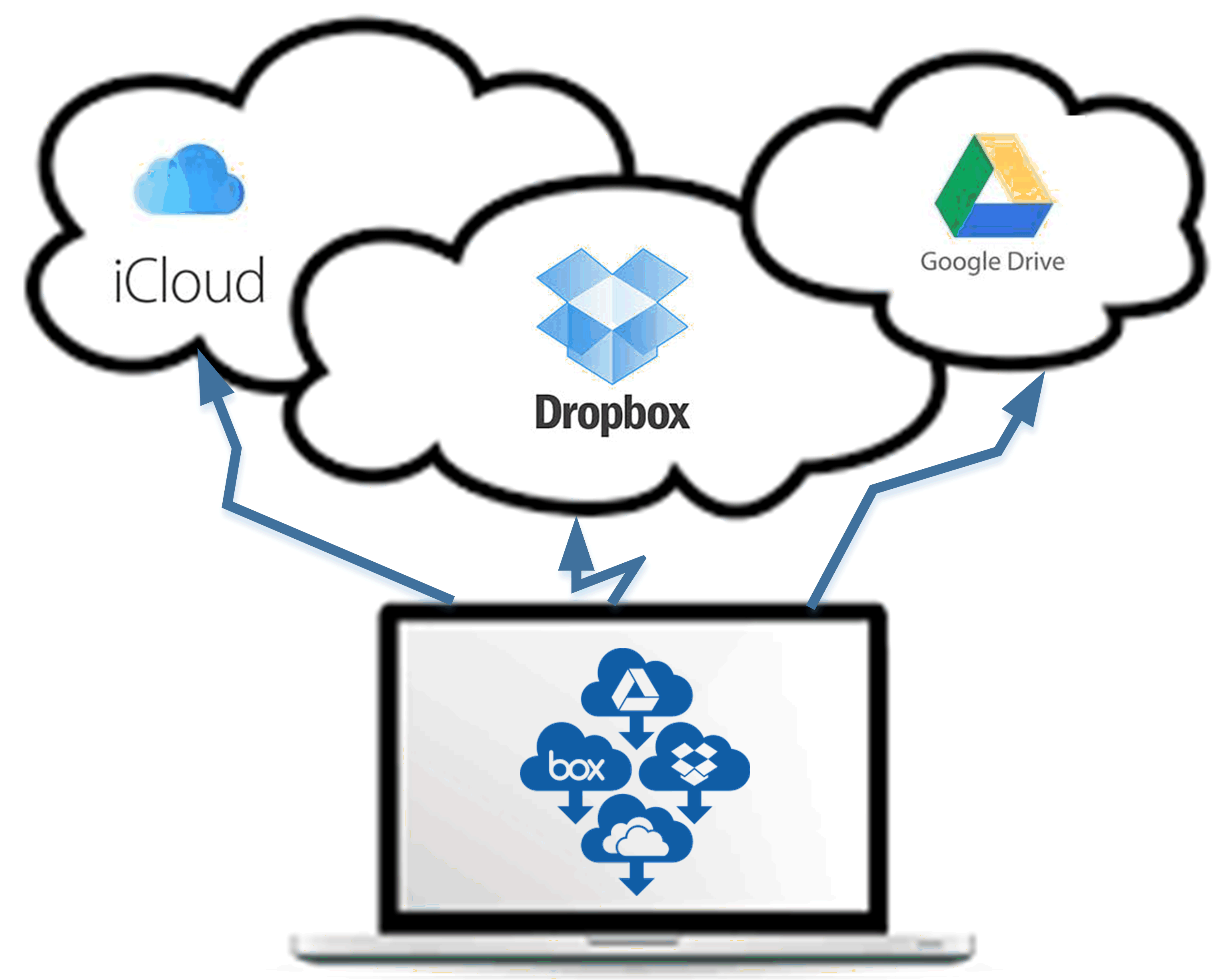}
\caption{Individual Secure Cloud Storage, with various cloud storage providers. The source encodes the original data using the individual security coding scheme suggested and then uploads $w$ encoded packets to different $\lceil k/w \rceil$ cloud storage providers, such that, each provider not only will not be able to decode the original data, but will also have zero information regarding any of the $k-w$ stored messages individually.}
\label{figure:cloud}
\end{figure}

It is important to note that utilizing the individual security coding scheme suggested in this paper, one not only ensures individual secrecy from potential eavesdroppers, but also can guarantee privacy from the hosting servers themselves, such that, each server not only will not be able to decode the original data but will have zero information regarding any of the stored message individually. For example, assume that in the example depicted in \Cref{figure:cloud}, the source $s$ (private user) wants to store a file $\textbf{M}$ in the cloud. To do that, the source can utilize $3$ different cloud storage providers, such as Google Drive, Microsoft OneDrive, Dropbox, etc. However, the source wants to keep the original information private. Hence, by encoding the original data using the individual security coding scheme suggested at the source, and then uploading at most $3$ encoded packets $\vec{Y}_{i_1},\ldots,\vec{Y}_{i_3}$ to any provider, the provider will store the packets in their servers $v$, but these will be kept ignorant of the original file.

\subsection{Wireless Networks}
The inherent broadcast nature of the wireless medium makes network coding techniques pertinent for wireless networks. Specifically, relying on network coding, instead of sending packets (unicast, multicast or broadcast packets) to each intended addressee individually, a source (or an intermediate node which needs to relay packets toward the destination) can transmit a manipulation (usually a linear combination) of the packets destined to the various receivers. A receiver collecting sufficient number of such combinations (coded packets) can reconstruct (decode) the original packets.
\ifdouble
\begin{figure}[h]
\centering
\includegraphics[trim= 0cm 0cm 0cm 0cm,clip,scale=0.85]{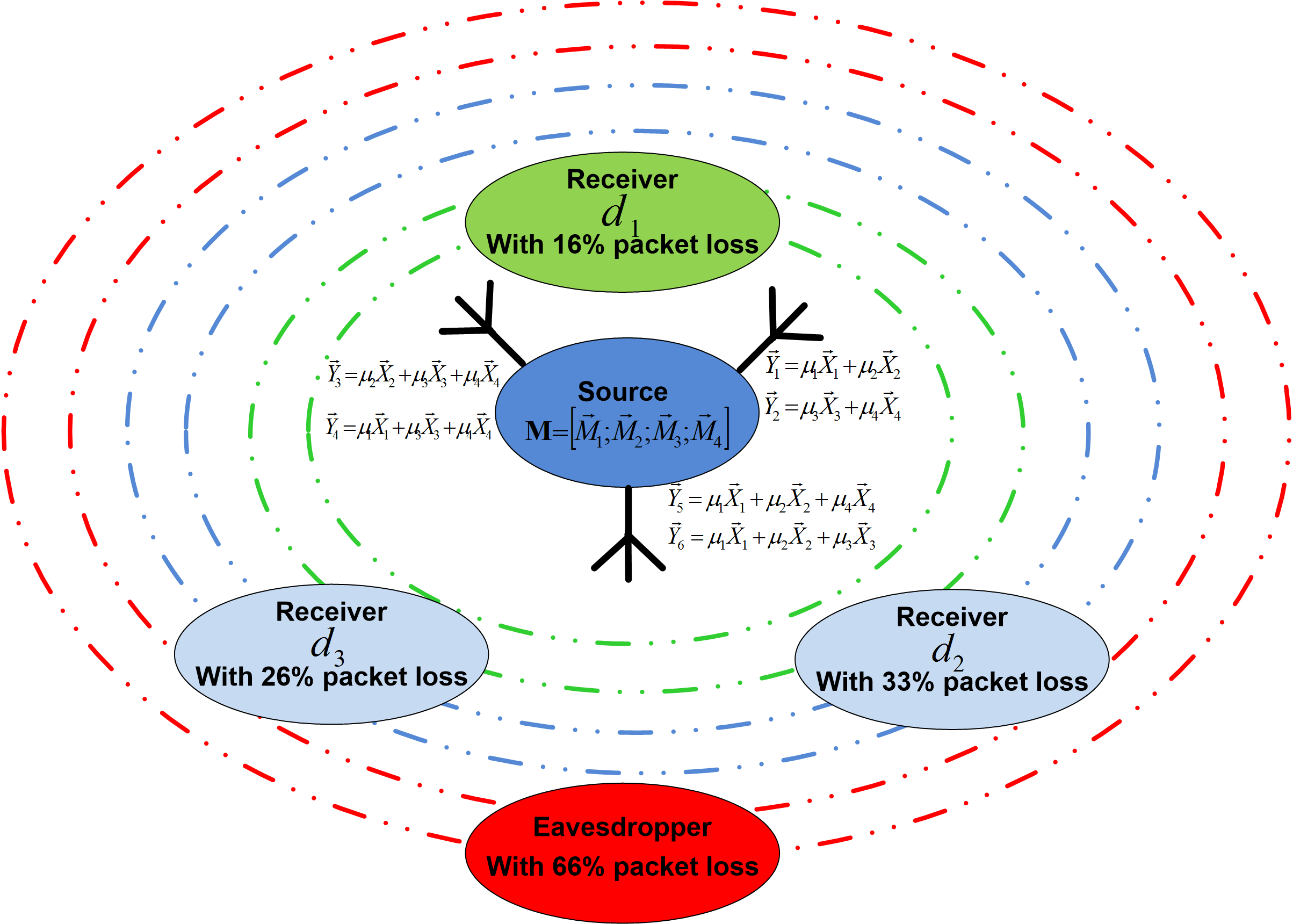}
\caption{Individual Secure Wireless Network. The source needs to disseminate $4$ message over a wireless network to $3$ legitimate users. In the individual secure model, we assume the existence of an eavesdropper. However, due to interference, collisions (low SINR) or low SNR, each of the receivers has a different packet loss rate, according to the physical constraints in the wireless networks.}
\label{figure:wireless_model}
\end{figure}
\else
\begin{figure}[h]
\centering
\includegraphics[trim= 0cm 0cm 0cm 0cm,clip,scale=0.85]{wireless_2.png}
\caption{Individual Secure Wireless Network. The source needs to disseminate $4$ message over a wireless network to $3$ legitimate users. In the individual secure model, we assume the existence of an eavesdropper. However, due to interference, collisions (low SINR) or low SNR, each of the receivers has a different packet loss rate, according to the physical constraints in the wireless networks.}
\label{figure:wireless_model}
\end{figure}
\fi
Relying on NC when the channel is lossy, i.e., there is a probability that a sent packet will not be received (decoded) by its intended receiver (receivers), has great advantages as instead of resending each un-coded packet until received correctly by its intended receiver, a sender keeps sending combinations of the original packets until each receiver collects a sufficient number of combinations (e.g., \cite{lun2005efficient,kim2011algebraic,popa2011going,talooki2015security,speidel2015can,lun2008coding,hansen2015network}). Accordingly, a sender can \emph{a priori} estimate the number of coded packets needed according to the most lossy channel and send coded packets accordingly, without relying on any feedbacks mechanism.

The secured version of this data dissemination problem requires that an eavesdropper with a degraded channel which can obtain only a subset of the transmitted packet will not be able to attain any information regarding any of the original packets. Utilizing the individual security coding scheme suggested in this paper, in which the source estimates the number of packets needed to be sent according to the estimated packet loss to each receiver, encodes the messages before the wireless transmission according to the procedure presented in \Cref{SecureGossipAlgorithmMultiplex} and the anticipated packet loss to the eavesdropper and broadcast the coded packets ensures that the legitimate users will be available to obtain the original transmitted data while any eavesdropper with higher packet loss rate will be kept ignorant. A simple illustration is given in \Cref{figure:wireless_model}: a transmitter utilizing MU-MIMO techniques to direct the beams toward its intended receivers such that eavesdroppers which are sparsely scattered are expected to experience a lower quality channel hence higher packet loss than the intended receivers; the transmitter is utilizing the individual-SMSM coding scheme suggested in \Cref{SecureGossipAlgorithmMultiplex}, ensuring individual security as proved in this paper.
In some sense, a similar application was suggested in \cite{yan2013algorithms,yan2014weakly} for secure data exchange,  where legitimate clients want to directly exchange information over a wireless channel in the presence of an eavesdropper. These works considered the matrix completion problem \cite{cohen1989ranks}, and provided an MRD code, such that, given constraints on the number of messages that each legitimate client has and transmits, established bounds on the number of transmissions over broadcasting channel required for both reliability at the legitimate clients (of all the data exchanged), and weakly secrecy at the eavesdropper who obtains all transmitted data.
\subsection{Live Broadcast of Video with Multi-Path Streaming}
Multi-Path routing techniques which enable the use of multiple alternative paths between a source and a destination through the network, has been widely exploited over the years to provide a variety of benefits such as load balancing, fault tolerance, bandwidth enhancement, etc. One such ubiquitous example is LiveU innovative solution for distributing live video streams via wireless networks \cite{liveu2017White,liveu2017site}. In these systems, the real-time recorded video is encoded in packets by the source. These encoded packets include pieces of the data to be transmitted through different distributed media. For example, the pieces of the data transmitted over various technologies such as cellular networks, WiFi, satellite, fiber internet, etc. or various providers, e.g., Sprint, T-Mobile, AT\&T Verizon, etc. A local server at the legitimate client decodes the data received from the different distributed media. This distributed streaming system maintains a high-quality viewer experience and cost-efficiency since the source can adapt the number of pieces dynamically to be transmitted by the different media. For example, if the connection using cellular or WiFi is lost during the real-time transmission, the source can route the pieces of the data dynamically by other connections or medias, taking into account the cost of each transmission by the optional connections.

In context to individual security suggested herein, we consider the case where there is an eavesdropper which has access to only a subset of the connections during the real-time distributed streaming (we assume that the eavesdropper can access any set of the streams unknown to the source, yet only a subset thereof). Utilizing the individual security coding scheme suggested in this paper, i.e., encoding the packets prior to the transmission, according to the coding scheme suggested in \Cref{SecureGossipAlgorithmMultiplex}, guarantees Individual Secure Live Broadcast of Video with Multi-Path Streaming, such that an eavesdropper which can capture at most $w$ streams transmitted over the different distributed medias is kept ignorant in the sense of having zero mutual information, regarding any set of $k_s$ messages individually, yet may potentially obtain \emph{insignificant} information about mixtures of packets transmitted.
Figure \ref{fig:LiveU} depicts a graphical representation of this system.
\ifdouble
\begin{figure}
\centering
\includegraphics[trim=0cm 0.0cm 0cm 0cm,clip,scale=0.85]{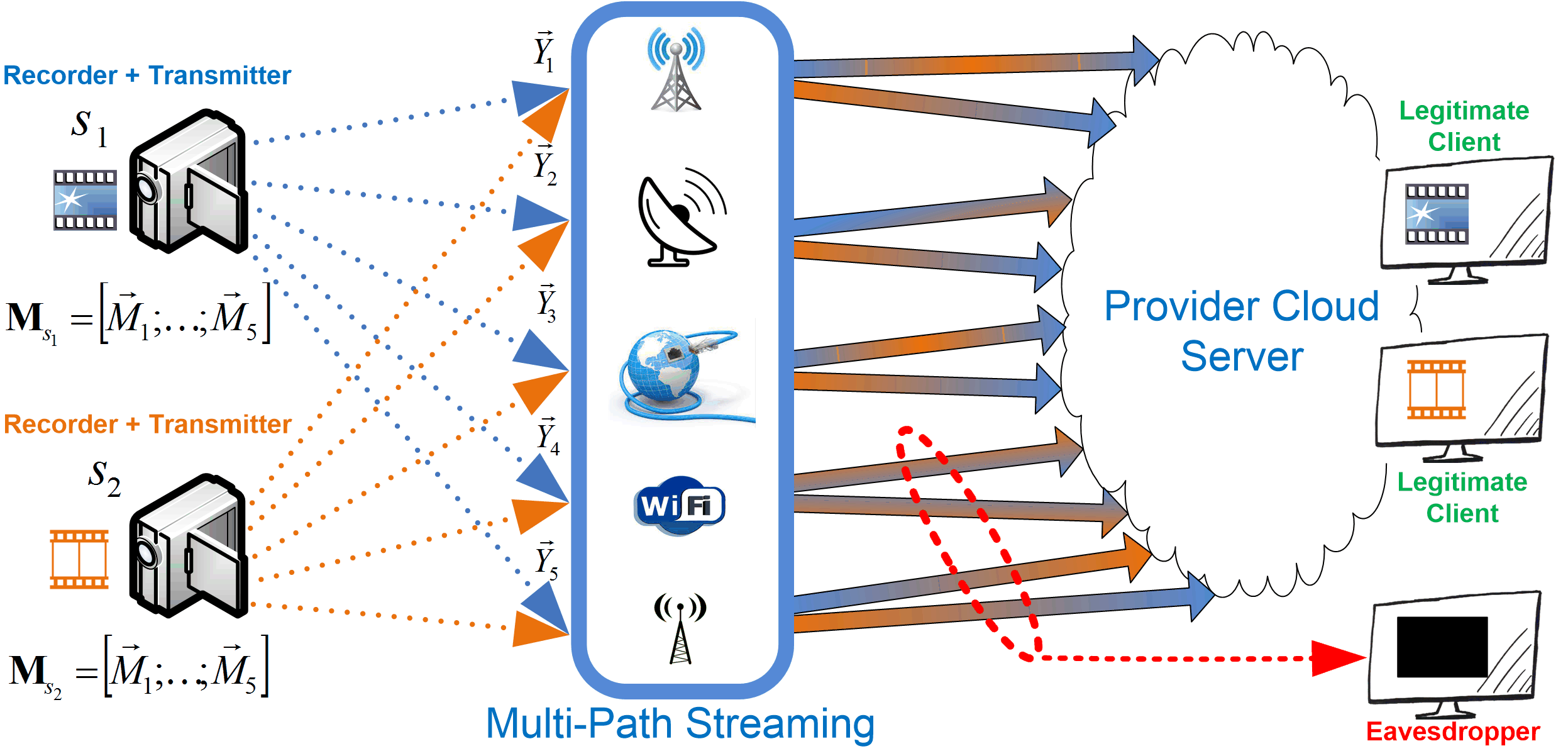}
\caption{Individual Secure Live Broadcast of Video with Multi-Path Streaming. The sources $s_1$, $s_2$ needs to transmit the real-time recorded video $\textbf{M}_{s_1}$, $\textbf{M}_{s_2}$,  respectively, encoded by LNC to $5$ packets $\vec{Y}_1,\ldots,\vec{Y}_5$ from each source, over the different medias.
The intermediate providers, such as cellular networks, WiFi, satellite, fiber internet, etc may use LNC before their routing transmission.
Then, the legitimate clints which received from a local provider cloud server all the packets, can decodes all the data. In the individual secure model of this problem, we assume the existence of an eavesdropper, which is able to obtain information from any $4$ connections. However, the individual secure code suggested herein, assure that the eavesdropper is not able to decode the original recorded information from the wiretapped connections.}
\label{fig:LiveU}
\end{figure}
\else
\begin{figure}[h]
\centering
\includegraphics[trim=0cm 0.0cm 0cm 0cm,clip,scale=1.0]{LiveU2.png}
\caption{Individual Secure Live Broadcast of Video with Multi-Path Streaming. The sources $s_1$, $s_2$ needs to transmit the real-time recorded video $\textbf{M}_{s_1}$, $\textbf{M}_{s_2}$,  respectively, encoded by LNC to $5$ packets $\vec{Y}_1,\ldots,\vec{Y}_5$ from each source, over the different medias.
The intermediate providers, such as cellular networks, WiFi, satellite, fiber internet, etc may use LNC before their routing transmission.
Then, the legitimate clints which received from a local provider cloud server all the packets, can decodes all the data. In the individual secure model of this problem, we assume the existence of an eavesdropper, which is able to obtain information from any $4$ connections. However, the individual secure code suggested herein, assure that the eavesdropper is not able to decode the original recorded information from the wiretapped connections.}
\label{fig:LiveU}
\end{figure}
\fi

\section{Conclusions}\label{conc}
In this paper, we proposed SMSM codes under an \emph{Individual Security} constraint.
In this model, the eavesdropper is kept ignorant, in the sense of having zero mutual information regarding each message separately, yet may potentially obtain \emph{insignificant} information about mixtures of packets transmitted.
In fact, it ensures Eve is kept ignorant of any \emph{set of $k-w$ messages}. That is, guarantee zero mutual information, with respect to any set of $k-w$ messages.

We completely characterized the rate region for individually secure MSM. Specifically, we showed that secure communication is achievable up to the min-cut, that is, without any decrease in the rate or any message ``blow-up" by extra randomness.
Moreover, we provided a code for Strong-SMSM by extra randomness, i.e., requiring Eve's mutual information \emph{with all messages simultaneously} to be zero. While this included a rate loss, it is important to note that in the code suggested the alphabet size did not increase with the network parameters due to the strong-security constraint.

Finally, we showed a few examples out of many important applications, like data centers, wireless networks, gossip and live broadcasting of video, for which the individual security coding schemes suggested is applicable, and achieves the full capacity of these systems.

\bibliographystyle{IEEE}
\bibliography{Ref1,Ref2}
\begin{IEEEbiography}[{\includegraphics[width=1in,height=1.25in,clip,keepaspectratio]{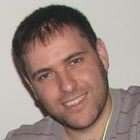}}]
{Alejandro Cohen} received the B.Sc. from the Department of Electrical Engineering, SCE college of engineering, Israel, in 2010
and M.Sc. degree at the communication system engineering, Ben-Gurion University of the Negev, Beer Sheva, Israel, in 2013.
Currently pursuing the Ph.D. degree in communication system engineering.
His main research interests are in the area of wireless communication, security, network information theory and network coding.
From 2007 to 2014 he was with DSP Group in Herzelya where he worked on voice enhancement and signal processing.
Currently he is with Intel in Petah-Tikva where he works in innovation group at mobile and wireless.
\end{IEEEbiography}

\begin{IEEEbiography}[{\includegraphics[width=1in,height=1.25in,clip,keepaspectratio]{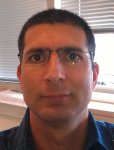}}]
{Asaf Cohen} is a senior lecturer at the Department of Communication Systems Engineering, Ben-Gurion University of the Negev, Israel. Before that, he was a post-doctoral scholar at the California Institute of Technology (Caltech). He received the B.Sc., M.Sc. (both with high honors) and Ph.D. from the Department of Electrical Engineering, Technion, Israel Institute of Technology, in 2001, 2003 and 2007, respectively. From 1998 to 2000 he was with the IBM Research Laboratory in Haifa where he worked on distributed computing. His areas of interest are information theory, learning andcoding. In particular, he is interested in sequential decision making, with applications to detection and estimation; Network security and anomaly detection; Network information theory and network coding; Statistical signal processing; Coding theory and performance analysis of codes. Dr. Cohen received several honors and awards, including the Viterbi post-doctoral scholarship, a student paper award at IEEE Israel 2006 and the Dr. Philip Marlin Prize for Computer Engineering, 2000.
\end{IEEEbiography}

\begin{IEEEbiography}[{\includegraphics[width=1in,height=1.25in,clip,keepaspectratio]{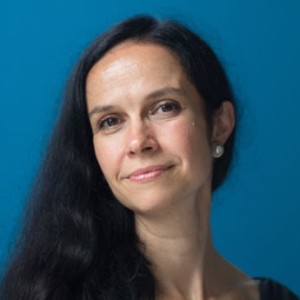}}]
{Muriel M\'{e}dard} is the Cecil H. Green Professor in the Electrical Engineering and Computer Science (EECS) Department at MIT and leads the Network Coding and Reliable Communications Group at the Research Laboratory for Electronics at MIT. She has co-founded three companies to commercialize network coding, CodeOn, Steinwurf and Chocolate Cloud. She has served as editor for many publications of the Institute of Electrical and Electronics Engineers (IEEE), of which she was elected Fellow, and she has served as Editor in Chief of the IEEE Journal on Selected Areas in Communications. She was President of the IEEE Information Theory Society in 2012, and served on its board of governors for eleven years. She has served as technical program committee co-chair of many of the major conferences in information theory, communications and networking. She received the 2009 IEEE Communication Society and Information Theory Society Joint Paper Award, the 2009 William R. Bennett Prize in the Field of Communications Networking, the 2002 IEEE Leon K. Kirchmayer Prize Paper Award, the 2018 ACM SIGCOMM Test of Time Paper Award and several conference paper awards. She was co-winner of the MIT 2004 Harold E. Edgerton Faculty Achievement Award, received the 2013 EECS Graduate Student Association Mentor Award and served as Housemaster for seven years. In 2007 she was named a Gilbreth Lecturer by the U.S. National Academy of Engineering. She received the 2016 IEEE Vehicular Technology James Evans Avant Garde Award, the 2017 Aaron Wyner Distinguished Service Award from the IEEE Information Theory Society and the 2017 IEEE Communications Society Edwin Howard Armstrong Achievement Award.
\end{IEEEbiography}

\begin{IEEEbiography}[{\includegraphics[width=1in,height=1.25in,clip,keepaspectratio]{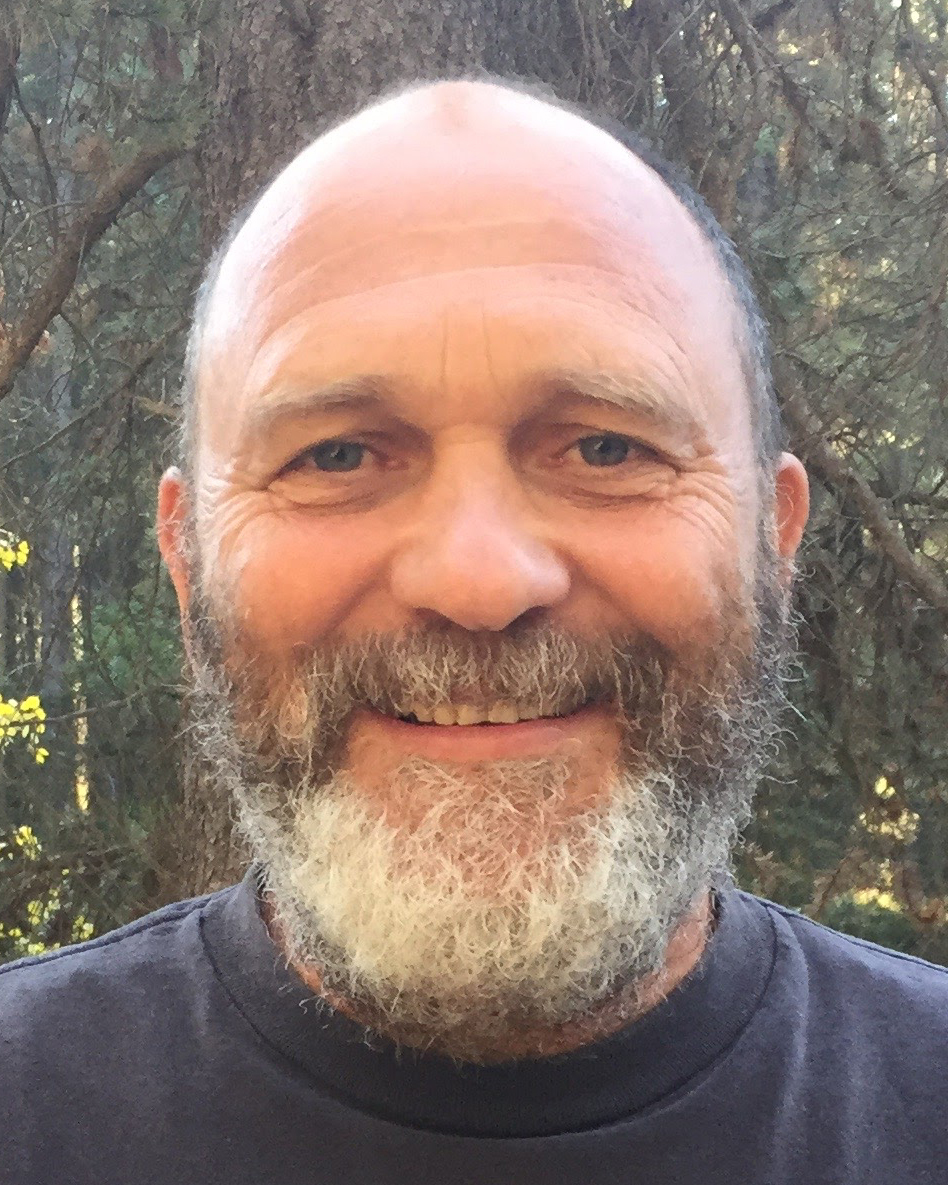}}]
{Omer Gurewitz} received the B.Sc. degree in Physics from Ben Gurion University, Beer Sheva, Israel, in 1991, and the M.Sc. and Ph.D. degrees in Electrical Engineering from the Technion Israel Institute of Technology,
Haifa, Israel, in 2000 and 2005, respectively. He is an Assistant Professor with the Department of Communication Systems Engineering, Ben Gurion University. Between 2005 and 2007, he was a Post-doctoral Researcher with the Electrical and Computer Engineering (ECE) Department, Rice University, Houston, TX, USA. His research interests are in the field of performance evaluation of wired and wireless communication networks. His current projects include cross-layer design and implementation of medium access protocols for next generation wireless communication.
\end{IEEEbiography}

\end{document}